\documentclass{aic}

\aicAUTHORdetails{%
  title = {An aperiodic set of 11 Wang tiles},
  author = {Emmanuel Jeandel and Micha\"el Rao},
  plaintextauthor = {Emmanuel Jeandel, Michael Rao}
}   

\aicEDITORdetails{%
   year={2021},
   number={1},
   received={14 March 2019},   
   revised={21 September 2020},    
   published={6 January 2021},  
   doi={10.19086/aic.18614},    
}   

\usepackage[utf8]{inputenc}
\usepackage{t1enc}
\usepackage{tikz}
\usepackage{ifthen}
\usepackage{amssymb}
\usepackage{amsmath}
\usepackage{amsthm}
\usepackage{caption}
\usepackage{subcaption}
\usepackage{multirow,bigdelim}

\newtheorem{theorem}{Theorem}
\newtheorem{conjecture}{Conjecture}
\newtheorem{lemme}{Lemma}
\newtheorem{proposition}{Proposition}
\newtheorem{fact}{Fact}
\newtheorem{corollary}{Corollary}

\usepackage[colorlinks=true]{hyperref}
\usetikzlibrary{shapes}
\usetikzlibrary{decorations.pathmorphing,decorations.pathreplacing,arrows,automata,patterns,positioning,shapes}

\begin{document}

\begin{frontmatter}[classification=text]

\title{An aperiodic set of 11 Wang tiles}

\author[ej]{Emmanuel Jeandel}
\author[mr]{Micha\"el Rao}

\begin{abstract}
We present a new aperiodic tileset containing 11 Wang tiles on 4 colors, and we show that this 
tileset is minimal, in the sense that no Wang set with either fewer than 11
tiles or fewer than 4 colors is aperiodic. This gives a definitive answer to the problem raised by Wang in 1961.
\end{abstract}
\end{frontmatter}


Wang tiles are square tiles with colored edges.
A tiling of the plane by Wang tiles consists of placing a Wang tile in
each cell of the grid $\mathbb{Z}^2$ so that contiguous edges share the
same color.
The formalism of Wang tiles was introduced  by Wang~\cite{wangpatternrecoII} to study decision procedures for a specific
fragment of logic (see Section~\ref{ssec:h1} for details).

Wang posed the question of whether an aperiodic tileset exists: a
set of Wang tiles which tiles the plane but cannot do so periodically.
His student Berger quickly gave an example of such an aperiodic tileset, but it consisted of a
very large number of tiles. 
The number of tiles needed was eventually reduced by Berger himself and then by others, to achieve an aperiodic set of only 13 Wang tiles in 1996
(see Section~\ref{ssec:h2} for an overview of previous aperiodic sets of Wang tiles). Their work in this apparently tedious exercise 
introduced several novel techniques to build aperiodic tilesets, and also to prove aperiodicity.

Other work has established that it is impossible to obtain an
aperiodic tileset with 4 tiles or less~\cite{GrSh}, and that it is also impossible to obtain aperiodic set of Wang tiles 
with fewer than 4 colors~\cite{tiles3}.

In this article, we fill all the gaps: we prove that there is an
aperiodic tileset with 11 Wang tiles and 4 colors, and we also prove that there is no
aperiodic tileset with fewer than 11 Wang tiles.

The discovery of this tileset, and the proof that there is no smaller aperiodic tileset,
was achieved by a computer search, which generated all possible candidates
with 11 tiles or less.

We proved that they were not aperiodic with 10 tiles or fewer.
Surprisingly, it was somewhat easy to do so for all of them except one.
The situation is different for 11 tiles: 
while we have found an aperiodic tileset, we also have a short list of tilesets which we are yet to characterize.
This computer search is described in Section \ref{sec:no10}, along with a result of independent interest:
we show that the tileset by Culik does not tile the plane if one tile is omitted. This section can be skipped by a reader who is only interested in our tileset itself. The tileset is presented in Section~\ref{sec:11}, and the
remaining sections prove that it is indeed an aperiodic tileset.

\section{Aperiodic sets of Wang tiles}

The following is a brief summary of the known aperiodic Wang tilesets, for which
further detail can be found in~\cite{GrSh}.

\subsection{Wang tiles and the \texorpdfstring{$\forall\exists\forall$}{AEA} problem}
\label{ssec:h1}
Wang tiles were introduced by Wang in 1961~\cite{wangpatternrecoII}, to study
the decidability of the $\forall\exists\forall$ fragment of first
order logic. In his article, Wang showed how to build a tileset $\tau$ and a subset
$\tau'\subseteq \tau$ so that there exists a tiling by $\tau$ of the upper
quadrant, with tiles in the first row in $\tau'$ if and only if  $\phi$ is
satisfiable. He builds this tileset starting from a
$\forall\exists\forall$ formula $\phi$.
In his development, the decidability of this particular tiling problem would imply that
the satisfiability of $\forall\exists\forall$ formulas was decidable.

More generally, Wang asked whether the more general tiling problem
(with no particular tiles in the first row) is decidable. He posed 
the \emph{fundamental conjecture}: every tileset either admits a
periodic tiling or it does not tile.

The following year, without having proven the conjecture above, Kahr, Moore
and Wang~\cite{KahrMooreWang} proved that the $\forall\exists\forall $ 
problem is indeed undecidable. They did so by reducing it to another tiling problem:
 we fix a subset $\tau'$ of tiles so that every tile on the
diagonal of the first quadrant is in $\tau'$.
This proof was later simplified by Hermes~\cite{Hermes,Hermes2}.
From the point of view of first order logic, the problem is therefore solved.
Formally speaking, the tiling problem with this diagonal constraint is reduced to a formula
of the form $\forall x \exists y \forall z \phi(x,y,z)$ where $\phi$
contains a binary predicate $P$ and occurrences of the subformula
$P(x,x)$, which code the diagonal constraint.
If we look at $\forall \exists\forall$ formulas that do not contain
the subformula $P(x,x)$ and $P(z,z)$, the decidability of this
particular fragment remained open. 

However, Berger proved a few years later~\cite{BergerPhD} both that the domino problem is
undecidable, and that an aperiodic tileset exists.
This implies that the fragment of
$\forall x\exists y\forall z$ where the only occurrence of the binary
predicates $P$ are of the form $P(x,z), P(y,z), P(z,y), P(z,x)$ is undecidable.

Over the years, other subcases of $\forall\exists\forall$ were described.
In 1975, Aanderaa and Lewis~\cite{Aanderaa} proved the undecidability of the fragment
of $\forall\exists\forall$ where the binary predicates $P$ can only appear in the form $P(x,z)$ and $P(z,y)$.
Their proof has the consequence that: the domino problem for
\emph{deterministic} tilesets is undecidable.

\subsection{Aperiodic tilesets}
\label{ssec:h2}
The first set of Wang tiles was provided by Berger in 1964.
The set contained in the 1966 AMS publication~\cite{Berger2}, has
20426 tiles, but Berger's original PhD Thesis~\cite{BergerPhD} also contains a simplified
version with 104 tiles. It should be noted that there is a mistake in Berger's paper: namely that 3
tiles are missing and 4 tiles are unneeded, bringing the 
actual tileset to 103 tiles.
This tileset is of a substitutive nature.
Knuth~\cite{KnuthWang} gave another simplified version of Berger's
original set, with 92 tiles (6 of which are actually unneeded,
bringing the number 86).

In 1966, Lauchli obtained an aperiodic set of 40 Wang tiles, which is published in a 1975 paper
by Wang~\cite{Wang2}.

In 1967, Robinson found an aperiodic set of 104 tiles, which was
mentioned only in a Notice of the AMS summary~\cite{robinson1967seven}.
Two simplifications of this
tileset exist: first, Poizat describes a tileset of 52 tiles~\cite{Poizat}.
Second, a tileset of 56 tiles was published by Robinson in 1971~\cite{Robinson} and is
probably his most well-known tileset. In that paper, he hints at a set of 35 Wang
tiles.

Later, Robinson managed to lower the number of tiles again to 32 using an idea
by Roger Penrose.
The same idea is used by Grunbaum and Shephard to obtain an aperiodic
set of 24 tiles~\cite{GrSh}.
In 1977, Robinson obtained a set of 24 tiles from a tiling method by Ammann.
For a long time the record was held by Ammann, who obtained, in 1978, 
a set of 16 Wang tiles.
When available, details on these tilesets are provided in~\cite{GrSh}.

In 1975, Aanderaa and Lewis~\cite{Aanderaa} build the first aperiodic \emph{deterministic} tileset. No details about the tileset are
provided but it is possible to extract one from the exposition by Lewis~\cite{Lewis}.
This construction was somehow forgotten in the literature, and the
first aperiodic deterministic tileset is usually attributed to Kari in
1992~\cite{KariNil}.

In 1989, Mozes showed a 
general method that can be used to translate any substitution tiling
into a set of Wang tiles~\cite{Mozes:1989}, which will be, of course, aperiodic.
There are multiple generalizations of this result (depending on the
exact definition of ``substitution tiling''), of which we cite
only a few~\cite{GS,OlFe,Gloa}. For a specific example, Socolar built
such a representation~\cite{Senechal} of the chair tiling, which, in
our vocabulary, can be done using 64 tiles.

The story stopped until 1996, when Kari invented a new method to build
aperiodic tileset, obtaining an aperiodic set of 14 tiles~\cite{Kari14}. This was reduced to 13 tiles by Culik~\cite{Culik}
using the same method. There was speculation that one of the 13 tiles was
unnecessary, and an unpublished manuscript by Kari and Culik hints at a method to show it. 
However, this is not true: the method developed
in the present article will show this is not the case.

In 1999, Kari and Papasoglu~\cite{KariPapa} presented the first 4-way
deterministic aperiodic set. The construction was later adapted by
Lukkarilla to provide a proof of undecidability of the 4-way domino
problem~\cite{Lukka}.

The construction by Robinson was later analyzed~\cite{Salon,AD,PuttingPieces,Savi} and simplified. In 2004, Durand, Levin, and Shen
presented~\cite{BLS2} a way to simplify exposition of proofs
of aperiodicity of such tilesets. Ollinger used this method in 2008 to
obtain an aperiodic tileset with 104 tiles~\cite{Ollinger}, which is
likely a rediscovery of the unpublished tileset by Robinson.
Other simplifications of Robinson constructions were given by Levin
in 2005~\cite{LevinAp} and Poupet in 2010 ~\cite{Poupet}, using ideas
similar to Robinson.

In 2008, Durand, Romashchenko, and Shen provided a new construction
based on the classical fixed point construction from computability
theory~\cite{BSR,DuRoSh}.

\def\ws{\mathcal{T}}
\def\ds{\mathcal{D}}
\def\Ds{{\mathcal{D}}}
\def\tr#1{{#1}^{\operatorname{tr}}}
\def\sc#1{\operatorname{s}({#1})}

\def\hiso{\sim}

\section{Preliminaries}

\subsection{Wang tiles}
A \emph{Wang tile} is a unit square with colored edges.
Formally, let $H,V$ be two finite sets (the horizontal and vertical
colors, respectively). A wang tile $t=(t_w,t_e,t_s,t_n)$ (for west, east, south and north) is an element of $H^2 \times V^2$.
Thus, the horizontal colors are used on west and east sides (that is, horizontal edges of the square), and vertical colors are used in north and south side (that is, vertical edges).

A \emph{Wang set} is a set of Wang tiles, formally viewed as a tuple
$(H,V,T)$, where $T\subseteq H^2\times V^2$ is the set of \emph{tiles}.
Figure~\ref{fig:ws13} presents a well-known example of a Wang set.
A Wang set is said to be \emph{empty} if $T=\emptyset$.

Let $\ws=(H,V,T)$ be a Wang set.
Let $X\subseteq \mathbb{Z}^2$.
A \emph{tiling of $X$ by $\ws$} is a mapping from $X$ to  $\ws$
so that contiguous edges have the same color; that is, 
it is a function $f : X \to T$ such that $f(x,y)_e=f(x+1,y)_w$ and
$f(x,y)_n=f(x,y+1)_s$ for every $(x,y)\in \mathbb{Z}^2$ when the function is defined.
We are especially interested in the tilings of $\mathbb{Z}^2$ by a
Wang set $\ws$. When we say a \emph{tiling of the plane by $\ws$}, or
simply a \emph{tiling by $\ws$}, we mean a tiling of $\mathbb{Z}^2$ by $\ws$.

A tiling $f$ is \emph{periodic} if there is a $(u,v) \in \mathbb{Z}^2 \setminus (0,0)$ 
such that $f(x,y)=f(x+u,y+v)$ for every $(x,y)\in \mathbb{Z}^2$. A
tiling is \emph{aperiodic} if it is not periodic.

A Wang set \emph{tiles $X$} (resp. \emph{tiles the plane}) if there exists a tiling of $X$ (resp. the plane) by $\ws$.
A Wang set is \emph{finite} if there is no tiling of the plane by $\ws$.
A Wang set is \emph{periodic} if there is a tiling $t$ by $\ws$ which is periodic.
A Wang set is \emph{aperiodic} if it tiles the plane, and every tiling by $\ws$ is not periodic. 

To quote a few well-known folklore results:
\begin{lemme}\label{lm:per}
If $\ws$ is periodic, then there is a tiling $t$ by $\ws$ with two
linearly independent translation vectors (in particular a tiling $t$
with  vertical and horizontal translation vectors).
\end{lemme}
\begin{lemme}\label{lm:comp}
If, for every $k\in \mathbb{N}$, there exists a tiling of $[0,\ldots, k]\times [0,\ldots, k]$ by $\ws$, then $\ws$ tiles the plane.
\end{lemme}

\subsection{Transducers}
One of the simplest but most crucial observations we will use in this article is that a Wang set may be viewed as a finite state transducer: a finite state automata with an input tape and an output tape.
In the entire paper, we use the same notation for Wang sets and for transducers, that is, a \emph{transducer} is a triplet $(H,V,T)$, where $H$ is the set of states, $V$ is the (input and output) alphabet, and $T$ is the set of transitions. 
Each $t=(w,e,s,n)\in T$ is a transition (in other words, a tile in the Wang set formalism), and the transducer authorizes the transition from 
the state $w$ to the state $e$, reading the letter $s$ on the input tape and writing $n$ on the output state.
It should be noted that, unlike usual automata and finite state transducers, we do not have either initial or final states: we work on biinfinite words.

Figure~\ref{fig:ws13} presents the popular set of Wang
tiles introduced by Culik from both points of view.
Note that we choose to label transitions with $s | n$ instead of, for example, $\frac{n}{s}$. This choice is intended to adhere to the classical way of depicting finite state transducers. 

\begin{figure}[htbp]
\begin{minipage}{.9\linewidth}
\center
\begin{tikzpicture}[scale=1.3]
\tikzstyle{every state}=[inner sep=.5mm]
\tikzstyle{tr}=[inner sep=.5mm]
\tikzstyle{trans}=[->,thick,auto,swap]
\node[state] (nm1) at (2.912000,5.600000) {$1$};
\node[state] (nm2) at (4.032000,4.256000) {$0$};
\node[state] (n0) at (1.792000,4.256000) {$2$};
\node[state] (n0p) at (5.00000,4.256000) {$0'$};
\node[state] (n12) at (5.00000,5.600000) {$\frac{1}{2}$};
\path[trans] (nm2) edge[bend right=25] node[tr] {$1|2$} (nm1);
\path[trans] (nm2) edge[bend right=0] node[tr] {$1|1$} (n0);
\path[trans] (nm1) edge[bend right=25] node[tr] {$1|2$} (n0);
\path[trans] (nm1) edge[bend right=0] node[tr] {$0|1$} (nm2);
\path[trans] (n0) edge[bend right=25] node[tr] {$0|2$} (nm2);
\path[trans] (n0) edge[bend right=0] node[tr] {$0|1$} (nm1);
\path[trans] (n0p)[loop right] edge node[tr,align=right] {$0'|0$\\$2|1$} (n0p);
\path[trans] (n0p) edge[bend right=15] node[tr,align=left] {$1|0$\\$1|0'$} (n12);
\path[trans] (n12) edge[loop right] node[tr,align=right] {$0'|0$\\$2|1$} (n12);
\path[trans] (n12) edge[bend right=15] node[tr] {$1|1$} (n0p);
 \end{tikzpicture}
\end{minipage}

\center
\begin{tikzpicture}[scale=0.9]
\tikzstyle{every node}=[font=\footnotesize]
\draw (0.000000,0.000000) -- (1.000000,0.000000) ;
\draw (0.000000,0.000000) -- (1.000000,1.000000) ;
\draw (0.000000,0.000000) -- (0.000000,1.000000) ;
\draw (1.000000,0.000000) -- (1.000000,1.000000) ;
\draw (0.000000,1.000000) -- (1.000000,1.000000) ;
\draw (0.000000,1.000000) -- (1.000000,0.000000) ;
\draw (0.430000,0.500000) node[left]{$0$} ;
\draw (0.570000,0.500000) node[right]{$1$} ;
\draw (0.500000,0.540000) node[above]{$2$} ;
\draw (0.500000,0.460000) node[below]{$1$} ;
\draw (1.300000,0.000000) -- (2.300000,0.000000) ;
\draw (1.300000,0.000000) -- (2.300000,1.000000) ;
\draw (1.300000,0.000000) -- (1.300000,1.000000) ;
\draw (2.300000,0.000000) -- (2.300000,1.000000) ;
\draw (1.300000,1.000000) -- (2.300000,1.000000) ;
\draw (1.300000,1.000000) -- (2.300000,0.000000) ;
\draw (1.730000,0.500000) node[left]{$0$} ;
\draw (1.870000,0.500000) node[right]{$2$} ;
\draw (1.800000,0.540000) node[above]{$1$} ;
\draw (1.800000,0.460000) node[below]{$1$} ;
\draw (2.600000,0.000000) -- (3.600000,0.000000) ;
\draw (2.600000,0.000000) -- (3.600000,1.000000) ;
\draw (2.600000,0.000000) -- (2.600000,1.000000) ;
\draw (3.600000,0.000000) -- (3.600000,1.000000) ;
\draw (2.600000,1.000000) -- (3.600000,1.000000) ;
\draw (2.600000,1.000000) -- (3.600000,0.000000) ;
\draw (3.030000,0.500000) node[left]{$1$} ;
\draw (3.170000,0.500000) node[right]{$2$} ;
\draw (3.100000,0.540000) node[above]{$2$} ;
\draw (3.100000,0.460000) node[below]{$1$} ;
\draw (3.900000,0.000000) -- (4.900000,0.000000) ;
\draw (3.900000,0.000000) -- (4.900000,1.000000) ;
\draw (3.900000,0.000000) -- (3.900000,1.000000) ;
\draw (4.900000,0.000000) -- (4.900000,1.000000) ;
\draw (3.900000,1.000000) -- (4.900000,1.000000) ;
\draw (3.900000,1.000000) -- (4.900000,0.000000) ;
\draw (4.330000,0.500000) node[left]{$1$} ;
\draw (4.470000,0.500000) node[right]{$0$} ;
\draw (4.400000,0.540000) node[above]{$1$} ;
\draw (4.400000,0.460000) node[below]{$0$} ;
\draw (5.200000,0.000000) -- (6.200000,0.000000) ;
\draw (5.200000,0.000000) -- (6.200000,1.000000) ;
\draw (5.200000,0.000000) -- (5.200000,1.000000) ;
\draw (6.200000,0.000000) -- (6.200000,1.000000) ;
\draw (5.200000,1.000000) -- (6.200000,1.000000) ;
\draw (5.200000,1.000000) -- (6.200000,0.000000) ;
\draw (5.630000,0.500000) node[left]{$2$} ;
\draw (5.770000,0.500000) node[right]{$0$} ;
\draw (5.700000,0.540000) node[above]{$2$} ;
\draw (5.700000,0.460000) node[below]{$0$} ;
\draw (6.500000,0.000000) -- (7.500000,0.000000) ;
\draw (6.500000,0.000000) -- (7.500000,1.000000) ;
\draw (6.500000,0.000000) -- (6.500000,1.000000) ;
\draw (7.500000,0.000000) -- (7.500000,1.000000) ;
\draw (6.500000,1.000000) -- (7.500000,1.000000) ;
\draw (6.500000,1.000000) -- (7.500000,0.000000) ;
\draw (6.930000,0.500000) node[left]{$2$} ;
\draw (7.070000,0.500000) node[right]{$1$} ;
\draw (7.000000,0.540000) node[above]{$1$} ;
\draw (7.000000,0.460000) node[below]{$0$} ;
\draw (0.000000,1.300000) -- (1.000000,1.300000) ;
\draw (0.000000,1.300000) -- (1.000000,2.300000) ;
\draw (0.000000,1.300000) -- (0.000000,2.300000) ;
\draw (1.000000,1.300000) -- (1.000000,2.300000) ;
\draw (0.000000,2.300000) -- (1.000000,2.300000) ;
\draw (0.000000,2.300000) -- (1.000000,1.300000) ;
\draw (0.480000,1.800000) node[left]{$0'$} ;
\draw (0.520000,1.800000) node[right]{$0'$} ;
\draw (0.500000,1.840000) node[above]{$0$} ;
\draw (0.500000,1.760000) node[below]{$0'$} ;
\draw (1.300000,1.300000) -- (2.300000,1.300000) ;
\draw (1.300000,1.300000) -- (2.300000,2.300000) ;
\draw (1.300000,1.300000) -- (1.300000,2.300000) ;
\draw (2.300000,1.300000) -- (2.300000,2.300000) ;
\draw (1.300000,2.300000) -- (2.300000,2.300000) ;
\draw (1.300000,2.300000) -- (2.300000,1.300000) ;
\draw (1.780000,1.800000) node[left]{$0'$} ;
\draw (1.820000,1.800000) node[right]{$0'$} ;
\draw (1.800000,1.840000) node[above]{$1$} ;
\draw (1.800000,1.760000) node[below]{$2$} ;
\draw (2.600000,1.300000) -- (3.600000,1.300000) ;
\draw (2.600000,1.300000) -- (3.600000,2.300000) ;
\draw (2.600000,1.300000) -- (2.600000,2.300000) ;
\draw (3.600000,1.300000) -- (3.600000,2.300000) ;
\draw (2.600000,2.300000) -- (3.600000,2.300000) ;
\draw (2.600000,2.300000) -- (3.600000,1.300000) ;
\draw (3.080000,1.800000) node[left]{$0'$} ;
\draw (3.170000,1.800000) node[right]{$\frac{1}{2}$} ;
\draw (3.100000,1.840000) node[above]{$0$} ;
\draw (3.100000,1.760000) node[below]{$1$} ;
\draw (3.900000,1.300000) -- (4.900000,1.300000) ;
\draw (3.900000,1.300000) -- (4.900000,2.300000) ;
\draw (3.900000,1.300000) -- (3.900000,2.300000) ;
\draw (4.900000,1.300000) -- (4.900000,2.300000) ;
\draw (3.900000,2.300000) -- (4.900000,2.300000) ;
\draw (3.900000,2.300000) -- (4.900000,1.300000) ;
\draw (4.380000,1.800000) node[left]{$0'$} ;
\draw (4.470000,1.800000) node[right]{$\frac{1}{2}$} ;
\draw (4.400000,1.840000) node[above]{$0'$} ;
\draw (4.400000,1.760000) node[below]{$1$} ;
\draw (5.200000,1.300000) -- (6.200000,1.300000) ;
\draw (5.200000,1.300000) -- (6.200000,2.300000) ;
\draw (5.200000,1.300000) -- (5.200000,2.300000) ;
\draw (6.200000,1.300000) -- (6.200000,2.300000) ;
\draw (5.200000,2.300000) -- (6.200000,2.300000) ;
\draw (5.200000,2.300000) -- (6.200000,1.300000) ;
\draw (5.630000,1.800000) node[left]{$\frac{1}{2}$} ;
\draw (5.770000,1.800000) node[right]{$\frac{1}{2}$} ;
\draw (5.700000,1.840000) node[above]{$0$} ;
\draw (5.700000,1.760000) node[below]{$0'$} ;
\draw (6.500000,1.300000) -- (7.500000,1.300000) ;
\draw (6.500000,1.300000) -- (7.500000,2.300000) ;
\draw (6.500000,1.300000) -- (6.500000,2.300000) ;
\draw (7.500000,1.300000) -- (7.500000,2.300000) ;
\draw (6.500000,2.300000) -- (7.500000,2.300000) ;
\draw (6.500000,2.300000) -- (7.500000,1.300000) ;
\draw (6.930000,1.800000) node[left]{$\frac{1}{2}$} ;
\draw (7.070000,1.800000) node[right]{$\frac{1}{2}$} ;
\draw (7.000000,1.840000) node[above]{$1$} ;
\draw (7.000000,1.760000) node[below]{$2$} ;
\draw (7.800000,1.300000) -- (8.800000,1.300000) ;
\draw (7.800000,1.300000) -- (8.800000,2.300000) ;
\draw (7.800000,1.300000) -- (7.800000,2.300000) ;
\draw (8.800000,1.300000) -- (8.800000,2.300000) ;
\draw (7.800000,2.300000) -- (8.800000,2.300000) ;
\draw (7.800000,2.300000) -- (8.800000,1.300000) ;
\draw (8.230000,1.800000) node[left]{$\frac{1}{2}$} ;
\draw (8.320000,1.800000) node[right]{$0'$} ;
\draw (8.300000,1.840000) node[above]{$1$} ;
\draw (8.300000,1.760000) node[below]{$1$} ;
\end{tikzpicture}
 \caption{The aperiodic set of 13 tiles obtained by Culik from an idea by Kari: the transducer view and the tile view.}
\label{fig:ws13}
\end{figure}

A \emph{biinfinite word} (or \emph{biinfinite sequence}) on the alphabet $A$ is a sequence $(w_i)_{i\in \mathbb{Z}}$ such that, for every $i\in \mathbb{Z}$, $w_i\in A$.
If $w$ and $w'$ are biinfinite words over the alphabet $V$, we will
write $w \ws w'$ if $w'$ is the image of $w$ by the transducer. More formally, $w\ws w'$ if there is a biinfinite sequence $ (q_i)_{i\in \mathbb{Z}}$ of states such that for every $i\in \mathbb{Z}$, $(q_i,q_{i+1},w_i,w'_i) \in T$.
In Wang tile formalism, $w \ws w'$ if one can tile a row such that $w$ are the sequence of colors on south edges, and $w'$ the color on north edges.
The transducer is usually nondeterministic, so this is indeed a
partial relation, not a function.

A \emph{run} of a transducer $\ws$ is a (possibly infinite or biinfinite) sequence of biinfinite words $(w_i)_{i\in I}$ (where $I$ is an interval of $\mathbb{Z}$) such that, for all $\{i,i+1\} \subset I$, $w_i\ws w_{i+1}$.
In this formalism, tilings of the plane correspond exactly to biinfinite runs of
the transducer, and $(H,V,T)$ tiles the plane if and only if there exists a biinfinite run of $(H,V,T)$.
Note also that, by compactness, there exists a biinfinite run of $(H,V,T)$ if and only if there exists an infinite run of $(H,V,T)$.

The composition of Wang sets, seen as transducers, is straightforward:
let $\ws=(H,V,T)$ and $\ws'=(H,V',T')$ be two Wang sets. Then $\ws\circ \ws'$ is the Wang set $(H\times H',V,T'')$, where $$T''=\{((w,w'),(e,e'),s,n') : (w,e,s,n) \in T, (w',e',s',n')\in T' \text{ and } n=s'\}.$$
Let $\ws^k$, $k\in \mathbb{N}\setminus \{0\}$ be $\ws$ if $k=1$, $\ws^{k-1} \circ \ws$ otherwise.

A reformulation of the original question is as follows:
\begin{lemme}
	A Wang set $\ws$ is finite if there is no infinite run of the
	transducer $\ws$: there is no biinfinite sequence $(w_k)_{k \in \mathbb{N}}$ so that
	$w_k \ws w_{k+1}$ for all $k$.
	
	A Wang set $\ws$ is periodic if and only if  there exists a biinfinite word $w$ and a positive integer $k$ 
	so that $w \ws^k w$.		
\end{lemme}

We will also use the  following operations on tilesets (or transducers):

\begin{description}
	\item[rotation]
Let $\tr{\ws}$ be $(V,H,T')$ where $T'=\{(s,n,e,w): (w,e,s,n)\in T\}$.
This operation corresponds to a rotation of the tileset by $90$ degrees.
\item[simplification]
Let $\sc{\ws}$ be the operation that deletes from $\ws$ 
any tile that cannot be used in a tiling of a (biinfinite) line row by
$\ws$. From the point of view of transducers, this corresponds to
eliminating sources and sinks from $\ws$. In particular, $\sc{\ws}$ is empty if and only if there are no
biinfinite words $w,w'$ s.t. $w \ws w'$.
\item[union] $\ws \cup \ws'$ is the disjoint union of transducers
  $\ws$ and $\ws'$: we first rename the states of both transducers so
  that they are all different, and then we take the union of the
  transitions of both transducers.
  Thus $w (\ws \cup \ws') w'$ if and only if $w \ws w'$ or $w \ws' w'$.
\end{description}
\paragraph{Equivalence of Wang sets.}

Once Wang sets are seen as transducers, it is easy to
see that the problems under consideration do not actually depend on $\ws$,
but only on the relation induced by $\ws$:
We say that two Wang sets $\ws=(H,V,T)$ and $\ws'=(H',V,T')$ are \emph{equivalent}
if they are equivalent as relations. In other words: 
for every pair of biinfinite words $(w,w')$ over $V$, $w\ws w' \Leftrightarrow w\ws' w'$.

In the course of the following proofs and algorithms, it will be useful 
to switch between equivalent Wang sets (transducers), in particular by
trying to simplify the sets as much as possible. For example, we can 
apply the operator $\sc{\ws}$ to
trim the colors/states (and thus the tiles/transitions) that cannot
appear in a biinfinite row (e.g., sources/terminals of the transducer
seen as a graph), or reduce the size of the transducer by coalescing
``equivalent'' states.

There are a few algorithms to simplify Wang sets.
First, as our transducers are nothing but (nondeterministic) finite automata
over the alphabet $V \times V$, it is tempting to try to \emph{minimize} them.
However, state (or transition) minimization of nondeterministic
automata is PSPACE-complete~\cite{MS72}.
Another plausible strategy, building the minimal deterministic automaton, has also proven to be
inefficient in practice.
The algorithm we use is based on the notion of \emph{strong
bisimulation equivalence} (or bisimulation, for short) of
labeled transitions systems ~\cite{Kan,Paige,Valmari,Ilie}.

A \emph{simulation} on the transducer $(H,V,T)$
is a relation $R \subseteq H^2$ such that
for every $u,u',v \in H$ and $a,b\in V$ such that $(u,u')\in R$ and $(u,v,a,b) \in T$, there exists $v'\in V$ such that $(u',v',a,b)\in T$.
A \emph{bisimulation} is a relation $R$ such that $R$ and $R^{-1}$ are simulations.
The \emph{bisimilarity relation}, which is the largest possible bisimulation,
is an equivalence relation, can be computed in linear time~\cite{Paige}.
The computation of the bisimilarity relation can be thought of as the non-deterministic
equivalent of Hopcroft's~\cite{Hopcroft} classical minimization algorithm for deterministic
automata.
Note that if we collapse equivalence classes in the transducer, we obtain a new transducer
which is equivalent to the previous one.

Another interesting option to simplify a transducer is the {simulation relation}, but the best known algorithm to compute it is in $O(n'm)$ time~\cite{Cc2017FoundationFA}, with $n'$ the number of equivalence classes, and $m$ the number of transitions,
which makes it impractical to use on large transducers. More so in our case, which sees transducers of up to several billions of transitions. 

\section{There is no aperiodic Wang sets with 10 tiles or less}
\label{sec:no10}
In this section, we present a
computer-assisted proof that there is no aperiodic Wang set
with 10 tiles or less.
The computer program can be found here:~\cite{prog}.

The general idea of the algorithm is straightforward: generate all Wang sets
with 10 tiles or less, and test each one to see whether it is aperiodic.
This method presents two difficulties here: first, there are a large number of
Wang sets with 10 tiles: for maximum efficiency, we have to discard  as soon as possible Wang
sets that are provably not aperiodic.
We then have to test the remaining sets for aperiodicity. Because aperiodicity
is an undecidable problem, our algorithm will not succeed on
all Wang sets; the remaining sets will have to be examined by hand.
\subsection{Generating all Wang sets with 10 tiles or less}

According to the general principle above, we do not actually have to
generate all Wang sets: we can refrain from generating sets that we
know to be aperiodic.

Let $\ws$ be a Wang set.
We say that $\ws$ is minimally aperiodic if $\ws$ is aperiodic and no
proper subset of $\ws$ is aperiodic (that is no proper subset of $\ws$
tiles the plane).
We will introduce criteria proving that some Wang sets are not
minimally aperiodic, and thus that we do not need to test them.

The key idea is to look at the graph $G$ underlying the transducer, that is, the transducer in which we neglect the labels of transitions.
Note that this is actually a \emph{multigraph}: there might be
multiple edges (transitions) joining two given vertices (states), and
there might also be self-loops.

This approach was introduced in~\cite{rolin2012}, and the following
lemma is more or less implicit in this article: 

\begin{lemme}
\label{lm:graph}
	Let $\ws$ be a Wang set, and $G$ the corresponding graph.
	\begin{itemize}
		\item Suppose there exist two vertices/states/colors $u,v \in G$ 
		  so that there is an edge
		  (hence a tile/transition) from  $u$ to $v$ and no path from $v$ to $u$.
		  Then $\ws$ is not minimal aperiodic.
		\item Suppose $G$ contains a strongly connected component
		  which is a cycle. Then $\ws$ is not minimal aperiodic.
\item If $G$ has only one vertex, then $\ws$ is not aperiodic.
		\item If the difference between the number of edges and the
		  number of vertices in $G$ is less than 2, then $\ws$ is not
		  minimal aperiodic.
	\end{itemize}		
\end{lemme}
\begin{proof}
In terms of tiles, the first case corresponds to a tile $t$ which can
appear at most once in each row. If $\ws$ tiles the plane, $\ws$ tiles
arbitrarily large regions without using the tile $t$. By compactness
(Lemma \ref{lm:comp}), $\ws \setminus \{t\}$ tiles the plane.

For the second case, suppose such a component exists.
This means there exist some tiles $ S\subseteq \ws$ so that
every time one of the tiles in $S$  appears, then the whole row is
periodic (of period the size of the cycle).
If $\ws$ is aperiodic, we cannot have a tiling where tiles of $S$
appear in two different rows, as we could deduce from it a periodic
tiling.
As a consequence, tiles from $S$ appear in at most one row, and using
the same compactness argument as before we deduce that $\ws \setminus
S$ tiles the plane.

For the third case, if $G$ has only one vertex and the Wang set tiles the plane, then one can construct a periodic tiling of the plane such that every column is the same column.

The proof of the fourth case can be found in~\cite{rolin2012}.
\end{proof}

We also suppose w.l.o.g. that there are no isolated vertices.
The number of graphs with the property of Lemma~\ref{lm:graph} are: 6 for 4 edges, 26 for 5 edges, 122 for 6 edges, 516 for 7 edges, 2517 for 8 edges, 13276 for 9 edges and 77809 for 10 edges.
The computer program \verb!gengraphs__N! generates the set of such graphs with \verb!N! edges.

This lemma gives a bird's-eye view of the program:
for a given $n \leq 10$, generate the set $\mathcal{G}$ of all
graphs with $n$ edges and at most $n-2$ vertices satisfying the
hypotheses of the lemma. Then for every $G_1$ and $G_2$ in $\mathcal{G}$,
we test all Wang sets for which the
first underlying graph (in west/east sides) is $G_1$,
and the underlying graph of $\tr{\ws}$ (that is, the north/south sides of $\ws$) is $G_2$.
To do so, we test every bijection between the edges of $G_1$ and the edges of $G_2$.
In terms of Wang tiles, a graph corresponds to
a specific assignment of colors to the east/west side: for this
particular assignment, we test all possible assignments of colors to
the north/south side.
The exact approach used in the software follows this principle, trying
as much as possible not to generate isomorphic tilesets.

\clearpage
\subsection{Testing Wang sets for aperiodicity}

In the previous section, we described how we generated Wang sets to
test. We now describe how we tested them for aperiodicity.

\subsubsection{Main program}

Recall that a Wang set is \emph{not} aperiodic if
\begin{itemize}
	\item either there exists $k$ so that $\sc{\ws^k}$ is empty: there are
	  no word $w,w'$ so that $w \ws^k w'$,
	\item or there exists $k$ so that $\ws^k$ is periodic: there exists
 	 a word $w$ so that $w \ws^k w$.
\end{itemize}	

The general algorithm to test for aperiodicity is therefore clear:
for each $k$, generate $\ws^k$, and then test whether one of the two cases
occurs. If it does, the set is not aperiodic. Otherwise, we go to the
next $k$. The algorithm stops when the computer program runs out of
memory. In that case, the algorithm was not able to decide if the
Wang set was aperiodic (it is after all an undecidable problem), and we
have to examine the Wang set carefully.

This approach works quite well in practice: when launched on a computer
with a reasonable amount of memory, it eliminates a very large number
of tilesets.
Of course, the key idea is to simplify $\ws^k$ as much as possible,
before computing $\ws^{k+1}$. 
Note that this should be done as fast as possible, as this will be
done for {all} Wang sets.
It turns out
that the easy simplification of deleting at each step the tiles that
cannot appear in a tiling of a row (i.e., vertices that are
sources/terminals) is already sufficient.

It is important to note that this approach, relying on transducers
(test whether the Wang set tiles $k$ consecutive rows, and whether it does
so periodically) turned out in practice to be much more efficient than
the naive approach of using tilings of squares (test whether the Wang set
tiles a square of size $k$, and whether it does so periodically).

\medskip

At this point, several possible improvements become apparent.
For example, the simplification of the transducers by bisimulation can be significant. 

However, we have to be careful about a few things. Firstly, some techniques can paradoxically waste more time than they save: the large majority of tilesets are quickly discarded by a simple and naive algorithm,
and the time spent on non-trivial cases represents only a tiny part of the overall time, even with this simple algorithm.
Secondly, these optimizations can make the program more difficult to read, to understand, and to check.

Where this is the case, we have chosen to opt for program clarity rather than computational efficiency:
without other improvements, one can show, in a reasonable time, that there is no aperiodic set of Wang tiles with at most 10 tiles.
For example, we do not even try to remove duplicate tiles, since this operation would require sorting the tiles.

\subsection{Computation}

Two independent programs~\cite{prog} are available in the folder \textit{src}, and another in the folder \textit{alternative}.
The result from both programs is the same:
we find only one hard case up to isomorphism, which is discussed in Section~\ref{hardest}.

We now give some details on the computation for the first program (the second is quite similar).
The computation for 10 tiles was done using the PSMN cluster (Pôle Scientifique de Modélisation Numérique) of the ENS de Lyon, and the computing resources of the LIP (Laboratoire d'Informatique du Parallélisme) of the ENS de Lyon, and required approximately 23 CPU years, that is roughly one week on 1000 cores.
For 9 tiles and less, the computations required approximately 38 CPU days.

For 10 tiles, there are $(77809\times77810)/2 \  \sim 3\times10^9$ cases. (By a case, we mean the test of all possible bijections between the edges of two graphs.) 
Most of the cases ($99.8\%$) take less than 1 second: the average time is 242ms and the median time is 155ms.
Except for the hardest case discussed below, the largest power we have to compute is $\ws^{126}$, and the largest transducer has $\sim18\times10^6$ transitions -- recall, however, that the program does not try to keep the transducer small.

It is difficult to recheck the result without substantial computing power.
As a kind of certificate, we provide all the hardest cases for 10 tiles: cases where either we have to compute at least $\ws^{30}$, or cases for which we get a transducer with at least $10^4$ edges.
We also give the hardest cases for sets of 5 up to 9 tiles.

\subsubsection{The hardest case}
\label{hardest}

Among all Wang sets, only 4 sets cannot be proven to be not aperiodic by the computer program.
All these 4 sets are isomorphic to the set $T_h$ presented in Figure~\ref{fig:10fail}.

It turned out that this particular Wang set is a special case of a
 general construction introduced by Kari~\cite{Kari14} of aperiodic Wang sets, save for the fact that 
a few tiles are missing. At this point, the situation could have
become desperate: it is not known whether Wang sets which were obtained by the method of
Kari less a few tiles actually tile the plane. In fact, the question was open: whether it was
possible to delete a tile from Culik's~\cite{Culik} 13 tileset to
obtain a set that still tiles the plane\footnote{You will find many experts on tilings that recollect this
  story wrongly and think that the (13) Wang set by Culik is the (14) Wang set
  from Kari with one tile removed. This is not the case. What happened
  is that there is one tile from the (13) Wang set by Culik that
  seemed likely to be unnecessary.}. It was conjectured by
both Kari and Culik that it was indeed possible.

\begin{figure}
	
\begin{center}
\begin{tikzpicture}[->,node distance=3cm,auto,initial text=]
        \tikzstyle{every state}=[minimum size=10mm]
\node[state] (A)                    {$1$};
\node[state] (B)  [right of=A]  {$0$};
\path (A) edge [loop above]   node[above] {$1|2$} (A);
\path (B) edge [loop above]   node[above] {$1|2$} (B);
\path (A) edge [bend left]   node[above] {$1|3, 0|1$} (B);
\path (B) edge [bend left]  node[below] {$2|3, 1|1$} (A);
\end{tikzpicture}
\end{center}
\begin{center}
\begin{tikzpicture}[->,node distance=3cm,auto,initial text=]
        \tikzstyle{every state}=[minimum size=10mm]
\node[state] (A)                    {$2/3$};
\node[state] (B)  [right of=A]  {$0'$};
\path (A) edge [loop above]   node[above] {$3|1$} (A);
\path (B) edge [loop above]   node[above] {$3|1$} (B);
\path (A) edge [bend left]   node[above] {$1|1$} (B);
\path (B) edge [bend left]  node[below] {$2|0$} (A);
\end{tikzpicture}
\end{center}
	\caption{The set $T_h$ of 10 tiles that comes very close, but fails to tile
	  the plane. It tiles a square of size $212 \times 212$}
\label{fig:10fail}	
\end{figure}	

We were able to prove that this tileset does not in fact tile
the plane.
Wang sets belonging to the family identified by Kari all work in the same way:
the biinfinite words that appear on each row can be thought of as reals, by taking the
average of all numbers (between $0$ and $3$ in our example)
that appear on the row.
Then, what the tileset does is apply a given piecewise affine
map to the real number.
In the case of our set of 10 tiles, the map $f$ is as follows:
\begin{itemize}
        \item if $1/2 \leq x < 3/2$, then $f(x) = 2x$,
        \item if $3/2 < x \leq 3$, then $f(x) = x /3$.
\end{itemize}

As can be seen from the first transducer, there cannot be two
consecutive 0 in $x$. This guarantees that $x \geq 1/2$, therefore $x \not= 0$, and, in particular, that this tileset has no periodic tiling.

If we used Kari's method to code this particular tileset,
the transducer that divides by $3$ would have $8$ tiles.
However, our particular set of 10 tiles does so with only $4$ tiles.
There is a way to explain how the division by $3$ works. First, we
treat it like a multiplication by $3$ by reversing the process.
Recall that the Beatty expansion of a real $x$ is given by 
$\beta_n(x)  = \left\lceil (n+1) x\right\rceil  - \left\lceil n x\right\rceil$.
Then one has \begin{fact}
        Let $0 < x \leq 1$ and define $b_n(x) = 2\beta_n(2x) - \beta_n(x)$.
        Then the second transducer transforms $(\beta_n)_{n \in \mathbb{N}}$ into
        $(b_n)_{n \in \mathbb{N}}$.
\end{fact}
Therefore, the second transducer multiplies by $3$ by doing $2\times 2
\times x - x$
somehow. It can be seen as a composition of a transducer that
transforms $(\beta_n)_{n \in \mathbb{N}}$ into $(\beta_n, b_n)_{n \in
  \mathbb{N}}$ (this can be done with only two states, using Kari's method)
and a transducer mapping each symbol $(x,y)$ into $2y-x$,
which can be done using only one state (this is just a relabeling).

There is no reason that doing the
transformation this way would work (in particular the equations given by Kari cannot
be applied to this particular transducer and prove that there is
a tiling of the plane), and indeed it doesn't: we were able to prove that
this particular Wang set does not tile the plane.

Once this tileset was identified as belonging to the
family of Kari tilesets, it is easy to see that, should it tile the
plane, it tiles a half plane starting from a word consisting only of
the symbol $3$.

\begin{fact}
If $T_h$ tiles the plane, then it tiles a half plane starting from a word consisting only of
the symbol $3$.
\end{fact}
\begin{proof}
  If one has a tiling of the plane, there is a biinfinite run $(w_i)$ of $T_h$.
  For $i,n\in \mathbb{N}$, let $x_{i,n}$ be the sum of the letters in $w_i$ (considering the letters as real numbers) between the positions $-n$ and $n$.
  The sequence $(x_{0,n}/(2n+1))_{n\in \mathbb{N}}$ is bounded, thus one can find an increasing function $\phi:\mathbb{N} \mapsto \mathbb{N}$ such that $x_{0,\phi(n)}/(2\phi(n)+1)$ converges. Let $y_0=\lim_{n\to \infty} x_{0,\phi(n)}/(2\phi(n)+1)$.
  Then $(x_{1,\phi(n)}/(2\phi(n)+1))_{n\in \mathbb{N}}$ converges, and by induction, for every $i\in \mathbb{N}$, $(x_{i,\phi(n)}/(2\phi(n)+1))_{n\in \mathbb{N}}$ converges. Let $y_i=\lim_{i\to \infty} x_{i,\phi(n)}/(2\phi(n)+1)$.

  Moreover, for every $i$, one has either $y_{i+1} = 2 y_i$ or $y_{i+1} = y_i/3$.
  One can suppose w.l.o.g. that for every $k\in\mathbb{N}$, $y_k\ne 3/2$. Otherwise, it is impossible that $y_{k'}=3/2$ with $k'>k$, and one can remplace the initial sequence by $(w_{i+k+1})_i)$.

  Thus, $y_i=f^i(y_0)$. Since $\log(2)/\log(3)$ is irrationnal, the set $\{y_i\}$ is dense, and for every $\epsilon >0$ there is an $i$ such that $|y_i-3| < \epsilon$.
  That is, for every $n\in \mathbb{N}$, the factor $3^n$ appears as a factor if one $w_i$.
  By compactness, there is a tiling of the half plane starting from a word consisting only of the symbol $3$.
\end{proof}

In our approach, we started from a transducer $\ws'$ which outputs a configuration
with only the symbol $3$, and built recursively $t_k = \ws' \ws^k$.

It turns out that $t_{31}$is empty, once reduced, which means that we cannot tile
$31$ consecutive rows starting from a word consisting only of $3$.
This fact is verified by the program \verb!hard10!.
Here, the naive approach --to remove sources and terminals-- takes too long for $t_{31}$, and we opted to use Tarjan's algorithm to find strongly connected components.

\begin{fact}
  $T_h$ does not tiles the plane.
\end{fact}

Thus, we get:

\begin{theorem}
	There is no aperiodic  Wang set with 10 tiles or less.
\end{theorem}	

\medskip

\sloppy
The fact that everything falls apart for $k = 31$ can be explained intuitively,
if we identify $([0.5,3]_{0.5\sim 3},\times)$ with the unit circle
$([0,1]_{0 \sim 1},+)$. What $f$ is doing is now just an addition
(modulo $1$) of $\frac{\log 2}{\log 2 + \log 3}$.
Now $31\frac{\log 2}{\log 2 + \log 3} = 11.992$ is near an integer, which means
that $\ws^{31}$ is ``almost'' the identity map. During the 30 first
steps, our map $\ws$ is able to deceive us, and it appears as if it would
tile the plane by using the degrees of freedom we have in the coding of the
reals. For $k = 31$, this is not possible anymore.

\fussy

Before removing unused transitions, $t_{31}$ contains a path of $212$
symbols $3$.
This means in particular that there exists a tiling of a rectangle of
size $212 \times 31$ where the top and the bottom sides are equal, thus
a tiling of a biinfinite vertical strip of width 212 by this tiling, and
thus a tiling of a square of size $212\times 212$.

\medskip

It turns out that the exact same method can be used for the set of 12
tiles obtained from Culik's set by removing one tile.
It corresponds to the same rotation, and we observe indeed the same
behavior: starting from a configuration of all $2$, it is not
possible to tile $31$ consecutive rows:
\begin{theorem}
The set of 13 tiles by Culik is minimal aperiodic: if any tile is
removed from this set, it does not tile the plane anymore.
\end{theorem}	
Note that the situation is still not well understood, and we
consider ourselves lucky to have obtained the result: first, we have to
execute the transducers in the right direction: $\ws' \ws^{-31}$ is
nonempty. Furthermore, the next step when $\ws^k$ returns near an
integer is for $k = 106$, and no computer, using our technique, has
enough memory to hope computing  $\ws^{106}$.

\begin{conjecture}
Every aperiodic tileset obtained by Kari's method is minimal
aperiodic.
\end{conjecture}	

\clearpage
\section{An aperiodic Wang set of 11 tiles - Proof sketch }\label{sec:11}

Using a similar method to the one presented in the last section, we were able to
enumerate and test sets of 11 tiles and find a few potential
candidates.
This computation took approximately one year on several hundred cores, using again the  PSMN cluster and the LIP cluster.

\medskip

Of these few candidates, three of them were extremely promising, and we
will prove that they are aperiodic sets.
These three sets look very similar, and the core of the proof of their aperiodicity is the same.

One of these three sets, $\ws'$ (Figure~\ref{fig:ws4}), uses only four colors, which is also minimal because no aperiodic set exists with only three colors~\cite{tiles3}.
To prove that $\ws'$ is aperiodic, we first show that $\ws$ (Figure~\ref{fig:ws}) is aperiodic, and then show that $\ws'$, which is a simple modification of $\ws$, is aperiodic.
The aperiodicity of the last set $\ws''$ (Figure~\ref{fig:ws4third}) is discussed in Section~\ref{sec:third}.

\begin{figure}[h!]
\center
\begin{tikzpicture}[scale=0.9]
\tikzstyle{every state}=[]
\tikzstyle{tr}=[inner sep=.5mm]
\tikzstyle{trans}=[->,thick,auto]
\node[state] (n0) at (1.600000,0.000000) {0};
\node[state] (n1) at (1.600000,3.200000) {1};
\node[state] (n2) at (3.200000,1.600000) {2};
\node[state] (n3) at (0.000000,1.600000) {3};
\path[trans] (n0) edge[loop right] node[tr] {$1|0$} (n0);
\path[trans] (n0) edge[bend left=10] node[tr] {$2|1$} (n3);
\path[trans] (n1) edge[bend left=10] node[tr] {$2|2$} (n0);
\path[trans] (n1) edge[loop right] node[tr] {$4|2$} (n1);
\path[trans] (n1) edge[bend left=10] node[tr] {$2|3$} (n3);
\path[trans] (n3) edge[bend left=10] node[tr] {$1|1$} (n0);
\path[trans] (n3) edge[bend left=10] node[tr,align=left] {$1|1$\\$2|2$} (n1);
\path[trans] (n3) edge[loop left] node[tr] {$3|1$} (n3);
\path[trans] (n2) edge[loop right] node[tr,align=left] {$1|4$\\$0|2$} (n2);
 \end{tikzpicture}

\medskip

\begin{tikzpicture}
\draw (0.000000,0.000000) -- (1.000000,0.000000) ;
\draw (0.000000,0.000000) -- (1.000000,1.000000) ;
\draw (0.000000,0.000000) -- (0.000000,1.000000) ;
\draw (1.000000,0.000000) -- (1.000000,1.000000) ;
\draw (0.000000,1.000000) -- (1.000000,1.000000) ;
\draw (0.000000,1.000000) -- (1.000000,0.000000) ;
\draw (0.430000,0.500000) node[left]{$0$} ;
\draw (0.570000,0.500000) node[right]{$0$} ;
\draw (0.500000,0.540000) node[above]{$0$} ;
\draw (0.500000,0.460000) node[below]{$1$} ;
\draw (1.300000,0.000000) -- (2.300000,0.000000) ;
\draw (1.300000,0.000000) -- (2.300000,1.000000) ;
\draw (1.300000,0.000000) -- (1.300000,1.000000) ;
\draw (2.300000,0.000000) -- (2.300000,1.000000) ;
\draw (1.300000,1.000000) -- (2.300000,1.000000) ;
\draw (1.300000,1.000000) -- (2.300000,0.000000) ;
\draw (1.730000,0.500000) node[left]{$0$} ;
\draw (1.870000,0.500000) node[right]{$3$} ;
\draw (1.800000,0.540000) node[above]{$1$} ;
\draw (1.800000,0.460000) node[below]{$2$} ;
\draw (2.600000,0.000000) -- (3.600000,0.000000) ;
\draw (2.600000,0.000000) -- (3.600000,1.000000) ;
\draw (2.600000,0.000000) -- (2.600000,1.000000) ;
\draw (3.600000,0.000000) -- (3.600000,1.000000) ;
\draw (2.600000,1.000000) -- (3.600000,1.000000) ;
\draw (2.600000,1.000000) -- (3.600000,0.000000) ;
\draw (3.030000,0.500000) node[left]{$1$} ;
\draw (3.170000,0.500000) node[right]{$0$} ;
\draw (3.100000,0.540000) node[above]{$2$} ;
\draw (3.100000,0.460000) node[below]{$2$} ;
\draw (3.900000,0.000000) -- (4.900000,0.000000) ;
\draw (3.900000,0.000000) -- (4.900000,1.000000) ;
\draw (3.900000,0.000000) -- (3.900000,1.000000) ;
\draw (4.900000,0.000000) -- (4.900000,1.000000) ;
\draw (3.900000,1.000000) -- (4.900000,1.000000) ;
\draw (3.900000,1.000000) -- (4.900000,0.000000) ;
\draw (4.330000,0.500000) node[left]{$1$} ;
\draw (4.470000,0.500000) node[right]{$1$} ;
\draw (4.400000,0.540000) node[above]{$2$} ;
\draw (4.400000,0.460000) node[below]{$4$} ;
\draw (5.200000,0.000000) -- (6.200000,0.000000) ;
\draw (5.200000,0.000000) -- (6.200000,1.000000) ;
\draw (5.200000,0.000000) -- (5.200000,1.000000) ;
\draw (6.200000,0.000000) -- (6.200000,1.000000) ;
\draw (5.200000,1.000000) -- (6.200000,1.000000) ;
\draw (5.200000,1.000000) -- (6.200000,0.000000) ;
\draw (5.630000,0.500000) node[left]{$1$} ;
\draw (5.770000,0.500000) node[right]{$3$} ;
\draw (5.700000,0.540000) node[above]{$3$} ;
\draw (5.700000,0.460000) node[below]{$2$} ;
\draw (6.500000,0.000000) -- (7.500000,0.000000) ;
\draw (6.500000,0.000000) -- (7.500000,1.000000) ;
\draw (6.500000,0.000000) -- (6.500000,1.000000) ;
\draw (7.500000,0.000000) -- (7.500000,1.000000) ;
\draw (6.500000,1.000000) -- (7.500000,1.000000) ;
\draw (6.500000,1.000000) -- (7.500000,0.000000) ;
\draw (6.930000,0.500000) node[left]{$3$} ;
\draw (7.070000,0.500000) node[right]{$0$} ;
\draw (7.000000,0.540000) node[above]{$1$} ;
\draw (7.000000,0.460000) node[below]{$1$} ;
\draw (0.000000,1.300000) -- (1.000000,1.300000) ;
\draw (0.000000,1.300000) -- (1.000000,2.300000) ;
\draw (0.000000,1.300000) -- (0.000000,2.300000) ;
\draw (1.000000,1.300000) -- (1.000000,2.300000) ;
\draw (0.000000,2.300000) -- (1.000000,2.300000) ;
\draw (0.000000,2.300000) -- (1.000000,1.300000) ;
\draw (0.430000,1.800000) node[left]{$3$} ;
\draw (0.570000,1.800000) node[right]{$1$} ;
\draw (0.500000,1.840000) node[above]{$1$} ;
\draw (0.500000,1.760000) node[below]{$1$} ;
\draw (1.300000,1.300000) -- (2.300000,1.300000) ;
\draw (1.300000,1.300000) -- (2.300000,2.300000) ;
\draw (1.300000,1.300000) -- (1.300000,2.300000) ;
\draw (2.300000,1.300000) -- (2.300000,2.300000) ;
\draw (1.300000,2.300000) -- (2.300000,2.300000) ;
\draw (1.300000,2.300000) -- (2.300000,1.300000) ;
\draw (1.730000,1.800000) node[left]{$3$} ;
\draw (1.870000,1.800000) node[right]{$1$} ;
\draw (1.800000,1.840000) node[above]{$2$} ;
\draw (1.800000,1.760000) node[below]{$2$} ;
\draw (2.600000,1.300000) -- (3.600000,1.300000) ;
\draw (2.600000,1.300000) -- (3.600000,2.300000) ;
\draw (2.600000,1.300000) -- (2.600000,2.300000) ;
\draw (3.600000,1.300000) -- (3.600000,2.300000) ;
\draw (2.600000,2.300000) -- (3.600000,2.300000) ;
\draw (2.600000,2.300000) -- (3.600000,1.300000) ;
\draw (3.030000,1.800000) node[left]{$3$} ;
\draw (3.170000,1.800000) node[right]{$3$} ;
\draw (3.100000,1.840000) node[above]{$1$} ;
\draw (3.100000,1.760000) node[below]{$3$} ;
\draw (3.900000,1.300000) -- (4.900000,1.300000) ;
\draw (3.900000,1.300000) -- (4.900000,2.300000) ;
\draw (3.900000,1.300000) -- (3.900000,2.300000) ;
\draw (4.900000,1.300000) -- (4.900000,2.300000) ;
\draw (3.900000,2.300000) -- (4.900000,2.300000) ;
\draw (3.900000,2.300000) -- (4.900000,1.300000) ;
\draw (4.330000,1.800000) node[left]{$2$} ;
\draw (4.470000,1.800000) node[right]{$2$} ;
\draw (4.400000,1.840000) node[above]{$4$} ;
\draw (4.400000,1.760000) node[below]{$1$} ;
\draw (5.200000,1.300000) -- (6.200000,1.300000) ;
\draw (5.200000,1.300000) -- (6.200000,2.300000) ;
\draw (5.200000,1.300000) -- (5.200000,2.300000) ;
\draw (6.200000,1.300000) -- (6.200000,2.300000) ;
\draw (5.200000,2.300000) -- (6.200000,2.300000) ;
\draw (5.200000,2.300000) -- (6.200000,1.300000) ;
\draw (5.630000,1.800000) node[left]{$2$} ;
\draw (5.770000,1.800000) node[right]{$2$} ;
\draw (5.700000,1.840000) node[above]{$2$} ;
\draw (5.700000,1.760000) node[below]{$0$} ;
\end{tikzpicture}
 
\caption{Wang set $\ws$.} \label{fig:ws}
\end{figure}

\begin{figure}[h!]
\center
\begin{tikzpicture}[scale=0.9]
\tikzstyle{every state}=[]
\tikzstyle{tr}=[inner sep=.5mm]
\tikzstyle{trans}=[->,thick,auto]
\node[state] (n0) at (1.600000,0.000000) {0};
\node[state] (n1) at (1.600000,3.200000) {1};
\node[state] (n2) at (3.200000,1.600000) {2};
\node[state] (n3) at (0.000000,1.600000) {3};
\path[trans] (n0) edge[loop right] node[tr] {$1|0$} (n0);
\path[trans] (n0) edge[bend left=10] node[tr] {$2|1$} (n3);
\path[trans] (n1) edge[bend left=10] node[tr] {$2|2$} (n0);
\path[trans] (n1) edge[loop right] node[tr] {$0|2$} (n1);
\path[trans] (n1) edge[bend left=10] node[tr] {$2|3$} (n3);
\path[trans] (n3) edge[bend left=10] node[tr] {$1|1$} (n0);
\path[trans] (n3) edge[bend left=10] node[tr,align=left] {$1|1$\\$2|2$} (n1);
\path[trans] (n3) edge[loop left] node[tr] {$3|1$} (n3);
\path[trans] (n2) edge[loop right] node[tr,align=left] {$1|0$\\$0|2$} (n2);
 \end{tikzpicture}

\medskip

\center
\begin{tikzpicture}
\draw (0.000000,0.000000) -- (1.000000,0.000000) ;
\draw (0.000000,0.000000) -- (1.000000,1.000000) ;
\draw (0.000000,0.000000) -- (0.000000,1.000000) ;
\draw (1.000000,0.000000) -- (1.000000,1.000000) ;
\draw (0.000000,1.000000) -- (1.000000,1.000000) ;
\draw (0.000000,1.000000) -- (1.000000,0.000000) ;
\draw (0.430000,0.500000) node[left]{$0$} ;
\draw (0.570000,0.500000) node[right]{$0$} ;
\draw (0.500000,0.540000) node[above]{$0$} ;
\draw (0.500000,0.460000) node[below]{$1$} ;
\draw (1.300000,0.000000) -- (2.300000,0.000000) ;
\draw (1.300000,0.000000) -- (2.300000,1.000000) ;
\draw (1.300000,0.000000) -- (1.300000,1.000000) ;
\draw (2.300000,0.000000) -- (2.300000,1.000000) ;
\draw (1.300000,1.000000) -- (2.300000,1.000000) ;
\draw (1.300000,1.000000) -- (2.300000,0.000000) ;
\draw (1.730000,0.500000) node[left]{$0$} ;
\draw (1.870000,0.500000) node[right]{$3$} ;
\draw (1.800000,0.540000) node[above]{$1$} ;
\draw (1.800000,0.460000) node[below]{$2$} ;
\draw (2.600000,0.000000) -- (3.600000,0.000000) ;
\draw (2.600000,0.000000) -- (3.600000,1.000000) ;
\draw (2.600000,0.000000) -- (2.600000,1.000000) ;
\draw (3.600000,0.000000) -- (3.600000,1.000000) ;
\draw (2.600000,1.000000) -- (3.600000,1.000000) ;
\draw (2.600000,1.000000) -- (3.600000,0.000000) ;
\draw (3.030000,0.500000) node[left]{$1$} ;
\draw (3.170000,0.500000) node[right]{$0$} ;
\draw (3.100000,0.540000) node[above]{$2$} ;
\draw (3.100000,0.460000) node[below]{$2$} ;
\draw (3.900000,0.000000) -- (4.900000,0.000000) ;
\draw (3.900000,0.000000) -- (4.900000,1.000000) ;
\draw (3.900000,0.000000) -- (3.900000,1.000000) ;
\draw (4.900000,0.000000) -- (4.900000,1.000000) ;
\draw (3.900000,1.000000) -- (4.900000,1.000000) ;
\draw (3.900000,1.000000) -- (4.900000,0.000000) ;
\draw (4.330000,0.500000) node[left]{$1$} ;
\draw (4.470000,0.500000) node[right]{$1$} ;
\draw (4.400000,0.540000) node[above]{$2$} ;
\draw (4.400000,0.460000) node[below]{$0$} ;
\draw (5.200000,0.000000) -- (6.200000,0.000000) ;
\draw (5.200000,0.000000) -- (6.200000,1.000000) ;
\draw (5.200000,0.000000) -- (5.200000,1.000000) ;
\draw (6.200000,0.000000) -- (6.200000,1.000000) ;
\draw (5.200000,1.000000) -- (6.200000,1.000000) ;
\draw (5.200000,1.000000) -- (6.200000,0.000000) ;
\draw (5.630000,0.500000) node[left]{$1$} ;
\draw (5.770000,0.500000) node[right]{$3$} ;
\draw (5.700000,0.540000) node[above]{$3$} ;
\draw (5.700000,0.460000) node[below]{$2$} ;
\draw (6.500000,0.000000) -- (7.500000,0.000000) ;
\draw (6.500000,0.000000) -- (7.500000,1.000000) ;
\draw (6.500000,0.000000) -- (6.500000,1.000000) ;
\draw (7.500000,0.000000) -- (7.500000,1.000000) ;
\draw (6.500000,1.000000) -- (7.500000,1.000000) ;
\draw (6.500000,1.000000) -- (7.500000,0.000000) ;
\draw (6.930000,0.500000) node[left]{$3$} ;
\draw (7.070000,0.500000) node[right]{$0$} ;
\draw (7.000000,0.540000) node[above]{$1$} ;
\draw (7.000000,0.460000) node[below]{$1$} ;
\draw (0.000000,1.300000) -- (1.000000,1.300000) ;
\draw (0.000000,1.300000) -- (1.000000,2.300000) ;
\draw (0.000000,1.300000) -- (0.000000,2.300000) ;
\draw (1.000000,1.300000) -- (1.000000,2.300000) ;
\draw (0.000000,2.300000) -- (1.000000,2.300000) ;
\draw (0.000000,2.300000) -- (1.000000,1.300000) ;
\draw (0.430000,1.800000) node[left]{$3$} ;
\draw (0.570000,1.800000) node[right]{$1$} ;
\draw (0.500000,1.840000) node[above]{$1$} ;
\draw (0.500000,1.760000) node[below]{$1$} ;
\draw (1.300000,1.300000) -- (2.300000,1.300000) ;
\draw (1.300000,1.300000) -- (2.300000,2.300000) ;
\draw (1.300000,1.300000) -- (1.300000,2.300000) ;
\draw (2.300000,1.300000) -- (2.300000,2.300000) ;
\draw (1.300000,2.300000) -- (2.300000,2.300000) ;
\draw (1.300000,2.300000) -- (2.300000,1.300000) ;
\draw (1.730000,1.800000) node[left]{$3$} ;
\draw (1.870000,1.800000) node[right]{$1$} ;
\draw (1.800000,1.840000) node[above]{$2$} ;
\draw (1.800000,1.760000) node[below]{$2$} ;
\draw (2.600000,1.300000) -- (3.600000,1.300000) ;
\draw (2.600000,1.300000) -- (3.600000,2.300000) ;
\draw (2.600000,1.300000) -- (2.600000,2.300000) ;
\draw (3.600000,1.300000) -- (3.600000,2.300000) ;
\draw (2.600000,2.300000) -- (3.600000,2.300000) ;
\draw (2.600000,2.300000) -- (3.600000,1.300000) ;
\draw (3.030000,1.800000) node[left]{$3$} ;
\draw (3.170000,1.800000) node[right]{$3$} ;
\draw (3.100000,1.840000) node[above]{$1$} ;
\draw (3.100000,1.760000) node[below]{$3$} ;
\draw (3.900000,1.300000) -- (4.900000,1.300000) ;
\draw (3.900000,1.300000) -- (4.900000,2.300000) ;
\draw (3.900000,1.300000) -- (3.900000,2.300000) ;
\draw (4.900000,1.300000) -- (4.900000,2.300000) ;
\draw (3.900000,2.300000) -- (4.900000,2.300000) ;
\draw (3.900000,2.300000) -- (4.900000,1.300000) ;
\draw (4.330000,1.800000) node[left]{$2$} ;
\draw (4.470000,1.800000) node[right]{$2$} ;
\draw (4.400000,1.840000) node[above]{$0$} ;
\draw (4.400000,1.760000) node[below]{$1$} ;
\draw (5.200000,1.300000) -- (6.200000,1.300000) ;
\draw (5.200000,1.300000) -- (6.200000,2.300000) ;
\draw (5.200000,1.300000) -- (5.200000,2.300000) ;
\draw (6.200000,1.300000) -- (6.200000,2.300000) ;
\draw (5.200000,2.300000) -- (6.200000,2.300000) ;
\draw (5.200000,2.300000) -- (6.200000,1.300000) ;
\draw (5.630000,1.800000) node[left]{$2$} ;
\draw (5.770000,1.800000) node[right]{$2$} ;
\draw (5.700000,1.840000) node[above]{$2$} ;
\draw (5.700000,1.760000) node[below]{$0$} ;
\end{tikzpicture}
 
\caption{Wang set $\ws'$, obtained from $\ws$ by collapsing the colors 4 and~0.} \label{fig:ws4}
\end{figure}

\begin{theorem}
  The Wang sets of Figure~\ref{fig:ws}, \ref{fig:ws4} and \ref{fig:ws4third} are aperiodic.
\end{theorem}
In this section, we sketch the proof of this result for the first
set $\ws$.

$\ws$ is the union of two Wang sets, $\ws_0$ and $\ws_1$, of 9 and 2 tiles respectively.
For $w\in \{0,1\}^*\setminus\{\epsilon\}$, let $\ws_{w} = {\ws_{w[1]} \circ \ws_{w[2]} \circ \ldots \ws_{w[\vert w\vert]}}$.

It can be seen by an easy computer verification that every tiling by $\ws$ can
be decomposed into a tiling by transducers
$\ws_1\ws_0\ws_0\ws_0\ws_0$ and  $\ws_1\ws_0\ws_0\ws_0$.

The simplifications of these two transducers, called $\ws_\texttt{a}$ and
$\ws_\texttt{b}$ will be obtained in Section~\ref{ssec:D}, and are depicted in Figure~\ref{fig:wsdd}.

We then study the transducer $\ws_D$ formed by the two transducers
$\ws_\texttt{a}$ and $\ws_\texttt{b}$ and prove that there exists a tiling by $\ws_D$,
and that any tiling by $\ws_D$ is aperiodic.
\begin{figure}[htbp]
\begin{minipage}{.9\linewidth}
\center
\begin{tikzpicture}[auto,scale=1]
\tikzstyle{every state}=[shape=ellipse,scale=0.7]
\tikzstyle{trans}=[->,line width=.5]
\tikzstyle{tr}=[inner sep=0.5mm,scale=0.7]
\node[state] (na) at (9.408000,4.928000) {a};
\node[state] (nb) at (8.512000,3.584000) {b};
\node[state] (nc) at (3.584000,4.928000) {c};
\node[state] (nd) at (4.480000,3.584000) {d};
\node[state] (ne) at (5.824000,3.584000) {e};
\node[state] (nf) at (7.182000,3.584000) {f};
\node[state] (ng) at (9.408000,2.240000) {g};
\node[state] (nh) at (3.584000,2.240000) {h};
\node[state] (ni) at (5.824000,4.928000) {i};
\node[state] (nj) at (7.168000,4.928000) {j};
\path[trans] (na) edge node[tr,swap] {$1|0$} (nb);
\path[trans] (nf) edge node[tr,swap] {$1|1$} (nb);
\path[trans] (ng) edge node[tr,swap] {$1|1$} (na);
\path[trans] (ne) edge node[tr,swap] {$1|1$} (nf);
\path[trans] (nc) edge node[tr,pos=.7,swap] {$1|0$} (nd);
\path[trans] (nb) edge[out=150,in=-20] node[tr,swap,pos=.15] {$1|0$} (nc);
\path[trans] (nb) edge node[tr,swap] {$1|1$} (ng);
\path[trans] (nd) edge node[tr,swap] {$1|1$} (ne);
\path[trans] (ni) edge node[tr] {$0|0$} (nj);
\path[trans] (nc) edge node[tr] {$0|0$} (ni);
\path[trans] (nd) edge node[tr] {$0|0$} (nh);
\path[trans] (nj) edge node[tr] {$0|0$} (na);
\path[trans] (nh) edge node[tr] {$0|0$} (nc);
\path[trans] (nh) edge node[tr,swap] {$0|1$} (ng);
 \end{tikzpicture}
\end{minipage}

\medskip

\begin{minipage}{.9\linewidth}
\center 
\begin{tikzpicture}[auto,scale=0.75]
\tikzstyle{every state}=[shape=ellipse,scale=0.7]
\tikzstyle{trans}=[->,line width=.5]
\tikzstyle{tr}=[inner sep=0.5mm,scale=0.7,swap]
\node[state] (nK) at (2.240000,6.720000) {K};
\node[state] (nL) at (2.240000,4.032000) {L};
\node[state] (nM) at (4.032000,4.032000) {M};
\node[state] (nM') at (4.032000,6.720000) {M'};
\node[state] (nN) at (5.824000,4.032000) {N};
\node[state] (nO) at (7.616000,4.032000) {O};
\node[state] (nO') at (7.616000,6.720000) {O'};
\node[state] (nP) at (9.408000,4.032000) {P};
\node[state] (nQ) at (9.408000,6.720000) {Q};
\node[state] (nR) at (5.824000,6.720000) {R};
\path[trans] (nP) edge node[tr,pos=.4] {$1|1$} (nR);
\path[trans] (nM') edge node[tr] {$1|1$} (nK);
\path[trans] (nR) edge node[tr,pos=.5,swap] {$1|1$} (nM);
\path[trans] (nQ) edge node[tr] {$0|0$} (nO');
\path[trans] (nN) edge node[tr] {$0|1$} (nO);
\path[trans] (nP) edge node[tr] {$0|1$} (nQ);
\path[trans] (nM) edge node[tr] {$0|1$} (nN);
\path[trans] (nR) edge node[tr,pos=.6] {$0|0$} (nL);
\path[trans] (nK) edge node[tr] {$0|1$} (nL);
\path[trans] (nO) edge node[tr,pos=.5,swap] {$0|0$} (nR);
\path[trans] (nO) edge node[tr] {$0|1$} (nP);
\path[trans] (nL) edge node[tr] {$0|1$} (nM);
\path[trans] (nO') edge node[tr] {$0|0$} (nR);
\path[trans] (nR) edge node[tr] {$1|1$} (nM');
 \end{tikzpicture}
\end{minipage}
\caption{$\ws_D$, the union of $\ws_\texttt{a}$ (top) and $\ws_\texttt{b}$ (bottom).}
\label{fig:wsdd}
\vspace{1cm}
\end{figure}

We will prove that the tileset is aperiodic, by proving that any tiling
is \emph{substitutive}. 
Let $u_{-2}= \epsilon$,$u_{-1} = \texttt{a}$, $u_0 = \texttt{b}, u_{n+2} =
u_n u_{n-1} u_n$.
For reference, here are the first values of $u$:
\[
\epsilon,\texttt{a},\texttt{b},\texttt{a}\texttt{a},\texttt{b}\texttt{a}\texttt{b},\texttt{a}\texttt{a}\texttt{b}\texttt{a}\texttt{a},\texttt{b}\texttt{a}\texttt{b}\texttt{a}\texttt{a}\texttt{b}\texttt{a}\texttt{b},\texttt{a}\texttt{a}\texttt{b}\texttt{a}\texttt{a}\texttt{b}\texttt{a}\texttt{b}\texttt{a}\texttt{a}\texttt{b}\texttt{a}\texttt{a},\texttt{b}\texttt{a}\texttt{b}\texttt{a}\texttt{a}\texttt{b}\texttt{a}\texttt{b}\texttt{a}\texttt{a}\texttt{b}\texttt{a}\texttt{a}\texttt{b}\texttt{a}\texttt{b}\texttt{a}\texttt{a}\texttt{b}\texttt{a}\texttt{b}
\]

\begin{figure}[htbp]
$T_n$ for $n$ odd:
\begin{center}
\begin{tikzpicture}[->,node distance=5cm,auto,initial text=]

        \tikzstyle{every state}=[minimum size=10mm]

\node[state] (A) at (0,0)                   {};
\node[state] (F) at (8,0){};
\path[ultra thick] (A) edge[bend left]   node[above] {$1^{g(n+2)-3}|0^{g(n+2)-3}$} (F);
\path[ultra thick] (F) edge[bend left]   node[below] {
$  \begin{array}{l@{|}l}
 1^{g(n+1)+3}&(110)0^{g(n+1)}\\
1^{g(n+3)+3}&0^{g(n+2)}(111)0^{g(n+1)}\\
1^{g(n+1)}(000)1^{g(n+2)}&0^{g(n+3)+3}\\
1^{g(n+1)}(100)&0^{g(n+1)+3}\\
1^{g(n+3)}(100)1^{g(n+1)}&0^{g(n+1)}(110)0^{g(n+3)}
\end{array}$} (A);
	
\end{tikzpicture}
\end{center}

$T_{n}$ for $n$ even:
\begin{center}
\begin{tikzpicture}[->,node distance=5cm,auto,initial text=]
        
        \tikzstyle{every state}=[minimum size=10mm]

\node[state] (A) at (0,0)                   {};
\node[state] (F) at (8,0) {};

\path[ultra thick] (A) edge[bend left]   node[above] {$0^{g(n+2)-3}|1^{g(n+2)-3}$} (F);
\path[ultra thick] (F) edge[bend left]   node[below] {
$  \begin{array}{l@{|}l}
0^{g(n+1)+3}&(100)1^{g(n+1)}\\
0^{g(n+3)+3}&1^{g(n+2)}(000)1^{g(n+1)}\\
0^{g(n+1)}(111)0^{g(n+2)}&1^{g(n+3)+3}\\
0^{g(n+1)}(110)&1^{g(n+1)+3}\\
0^{g(n+3)}(110)0^{g(n+1)}&1^{g(n+1)}(100)1^{g(n+3)}\\
\end{array}$} (A);
\end{tikzpicture}
\end{center}
\caption{The family of transducers $T_n$}
\label{fig:family}
\end{figure}

Let $g(n)$, $n\in\mathbb{N}$ be the $(n+1)$-th Fibonacci number, that is $g(0)=1$, $g(1)=2$ and $g(n+2)=g(n)+g(n+1)$ for every $n\in\mathbb{N}$.
Remark that $u_n$ is of size $g(n)$.
Then we will  prove:
\begin{proposition}
Any tiling of the plane by $\ws_D$ can be divided into strips of
vertical width $g(n), g(n+1)$ or $g(n+2)$ so that each strip is a tiling by
$\ws_{u_n},
          \ws_{u_{n+1}}$ or $\ws_{u_{n+2}}$.
        \end{proposition}  
Remark that, by definition,
$\ws_{u_{n+3}} = \ws_{u_{n+1}} \circ \ws_{u_{n}} \circ \ws_{u_{n+1}}$.

We will prove this by induction on $n$. For this, we introduce a
family of transducers, presented in Figure~\ref{fig:family}, and we will
prove the following:
\begin{itemize}
\item
We show in Proposition \ref{prop:010} that for any tiling of the plane by $\ws_D$, the words in each row avoid the factors $010$ and $101$.
  
  \item We prove (Section~\ref{ssec:seq}) that every tiling by $\ws_D =
	\ws_\texttt{a} \cup \ws_\texttt{b}$ can be seen as a tiling by $\ws_{u_0} \cup \ws_{u_1} \cup
	\ws_{u_2} = \ws_\texttt{b} \cup \ws_\texttt{aa} \cup \ws_{\texttt{bab}}$.
      \item We prove (Section~\ref{ssec:seq}) that for words $u,v \in W$,
        $u\ws_{u_i} v  \iff u T_i v$. This means that we can
        interchangeably replace the Wang sets $\ws_{u_0}, \ws_{u_1},
        \ws_{u_2}$ by $T_0, T_1, T_2$ without changing the tilings of
        the plane.
  \item At this point, it becomes obvious that $\ws$ is aperiodic if and only if the
    Wang set $T_0 \cup T_1 \cup T_2$ is aperiodic.    
  \item We prove (Section~\ref{sec:recur}) that $T_{n+3} = T_{n+1} \circ T_n \circ
		  T_{n+1}$ for all $n$. As $\ws_{u_{n+3}} =
                  \ws_{u_{n+1}} \circ \ws_{u_{n}} \circ
                  \ws_{u_{n+1}}$, we obtain by
                  an easy induction\footnote{To be rigorous, one also
                    needs to use that if $r \ws_{u_{n+1}} s
                    \ws_{u_{n}} t \ws_{u_{n+1}} v$ with $r,v \in W$,
                  then $s,t \in W$ which is a clear consequence of
                  Proposition \ref{prop:010}.}
                that for all $u,v \in W$,  $u \ws_{u_n} v$ if and only if $u T_n v$.
                  
                \item We then prove (Section~\ref{sec:end}) that
                  any tiling by $T_n, T_{n+1},$ $T_{n+2}$ can be
                  rewritten as a tiling by $T_{n+1}, T_{n+2},
                  T_{n+3}$.                  
                  As a consequence, any tiling by $\ws_{u_n},
          \ws_{u_{n+1}}$ and $\ws_{u_{n+2}}$ can be rewritten as a tiling
          by $\ws_{u_{n+1}},\ws_{u_{n+2}},\ws_{u_{n+3}}$, by replacing any
          block $\ws_{u_{n+1}}\ws_{u_n}\ws_{u_{n+1}}$ by $\ws_{u_{n+3}}$
          (the difficulty is to prove that by doing this, there is no
          remaining occurrence of $\ws_{u_n}$).
        \end{itemize}

        This proves the proposition and the theorem.

Finally, we explain in Section~\ref{sec:end} how the same proof gives
us also the aperiodicity of the set $\ws'$.

\section{From \texorpdfstring{$\ws$ to $\ws_D$ then to $T_0,T_1,T_2$}{T to TD then T0,T1,T2}}

\subsection{From \texorpdfstring{$\ws$ to $\ws_D$}{T to TD}}
\label{ssec:D}
Recall that our Wang set $\ws$ can be seen as the union of two Wang sets, $\ws_0$ and $\ws_1$, of 9 and 2 tiles respectively.

For $w\in \{0,1\}^*\setminus\{\epsilon\}$, let $\ws_{w} = {\ws_{w[1]} \circ \ws_{w[2]} \circ \ldots \ws_{w[\vert w\vert]}}$.
The following facts can be easily checked by computer or by hand:
\begin{fact}
The transducers $\sc{\ws_{11}}$, $\sc{\ws_{101}}$, $\sc{\ws_{1001}}$ and $\sc{\ws_{00000}}$ are empty.
\end{fact}

Thus, if $t$ is a tiling by $\ws$, then there exists a biinfinite binary word $w\in \{1000,10000\}^{\mathbb{Z}}$ such that $t(x,y)\in T(\ws_{w[y]})$ for every $x,y\in\mathbb{Z}$.
Let $\ws_A=\sc{\ws_{1000} \cup \ws_{10000}}$ (see Figure~\ref{fig:wsa}).
There is a bijection between the tilings by $\ws$ and the tilings by $\ws_A$, and $\ws$ is aperiodic if and only if $\ws_A$ is aperiodic.

\begin{figure}

\begin{subfigure}[b]{.999\linewidth}
\begin{minipage}{.44\linewidth}
\center
\begin{tikzpicture}[auto,scale=.75]
\tikzstyle{every state}=[shape=ellipse,inner sep=0mm,font=\scriptsize,scale=0.6]
\tikzstyle{trans}=[->,thin]
\tikzstyle{tr}=[inner sep=0.2mm,font=\footnotesize,scale=0.6]
\node[state] (n20330) at (5.824000,4.928000) {20330};
\node[state] (n21030) at (9.408000,4.928000) {21030};
\node[state] (n21033) at (7.182000,3.584000) {21033};
\node[state] (n21100) at (10.122000,3.066000) {21100};
\node[state] (n21103) at (9.408000,2.240000) {21103};
\node[state] (n21113) at (5.824000,3.584000) {21113};
\node[state] (n21130) at (3.584000,4.928000) {21130};
\node[state] (n21300) at (8.512000,4.032000) {21300};
\node[state] (n21310) at (8.512000,3.136000) {21310};
\node[state] (n21311) at (7.770000,1.652000) {21311};
\node[state] (n21330) at (4.480000,3.584000) {21330};
\node[state] (n23100) at (7.168000,4.928000) {23100};
\node[state] (n23300) at (3.584000,2.912000) {23300};
\node[state] (n23310) at (4.256000,2.240000) {23310};
\path[trans] (n21030) edge node[tr,swap] {$1|0$} (n21300);
\path[trans] (n21033) edge node[tr,swap] {$1|1$} (n21300);
\path[trans] (n21030) edge node[tr] {$1|0$} (n21310);
\path[trans] (n21033) edge node[tr,swap] {$1|1$} (n21310);
\path[trans] (n21033) edge node[tr,swap,pos=.3] {$1|1$} (n21311);
\path[trans] (n21100) edge node[tr,swap] {$1|0$} (n21030);
\path[trans] (n21103) edge node[tr] {$1|1$} (n21030);
\path[trans] (n21113) edge node[tr] {$1|1$} (n21033);
\path[trans] (n21130) edge node[tr] {$1|0$} (n21330);
\path[trans] (n21300) edge[bend left=5] node[tr] {$1|0$} (n21130);
\path[trans] (n21310) edge node[tr] {$1|1$} (n21103);
\path[trans] (n21311) edge[bend right=55,swap] node[tr] {$1|2$} (n21100);
\path[trans] (n21311) edge node[tr,swap] {$1|3$} (n21103);
\path[trans] (n21330) edge node[tr] {$1|1$} (n21113);
\path[trans] (n20330) edge node[tr] {$0|0$} (n23100);
\path[trans] (n21130) edge node[tr] {$0|0$} (n20330);
\path[trans] (n21330) edge node[tr] {$0|0$} (n23300);
\path[trans] (n21330) edge node[tr] {$0|0$} (n23310);
\path[trans] (n23100) edge node[tr] {$0|0$} (n21030);
\path[trans] (n23300) edge node[tr] {$0|0$} (n21130);
\path[trans] (n23310) edge node[tr,pos=.4] {$0|1$} (n21103);
 \end{tikzpicture}
\end{minipage}
\begin{minipage}{.55\linewidth}
\center 
\begin{tikzpicture}[auto,scale=0.55]
\tikzstyle{every state}=[shape=ellipse,inner sep=0mm,font=\scriptsize,scale=0.6]
\tikzstyle{trans}=[->,thin]
\tikzstyle{tr}=[inner sep=0.2mm,font=\footnotesize,scale=0.6]
\node[state] (n2030) at (9.632000,4.032000) {2030};
\node[state] (n2033) at (5.824000,4.032000) {2033};
\node[state] (n2100) at (10.752000,8.960000) {2100};
\node[state] (n2103) at (9.632000,6.720000) {2103};
\node[state] (n2110) at (1.120000,8.736000) {2110};
\node[state] (n2111) at (10.752000,2.688000) {2111};
\node[state] (n2113) at (4.032000,4.032000) {2113};
\node[state] (n2130) at (5.824000,6.720000) {2130};
\node[state] (n2131) at (3.990000,7.364000) {2131};
\node[state] (n2133) at (2.240000,4.032000) {2133};
\node[state] (n2300) at (7.392000,4.928000) {2300};
\node[state] (n2310) at (7.840000,3.584000) {2310};
\node[state] (n2311) at (9.632000,8.064000) {2311};
\node[state] (n2330) at (2.240000,6.720000) {2330};
\node[state] (n2331) at (1.176000,5.992000) {2331};
\path[trans] (n2100) edge[bend right=25,swap] node[tr] {$1|0$} (n2130);
\path[trans] (n2103) edge node[tr,pos=.7] {$1|1$} (n2130);
\path[trans] (n2103) edge[bend right=5,pos=.85,swap] node[tr] {$1|1$} (n2131);
\path[trans] (n2110) edge[bend left=10] node[tr] {$1|1$} (n2103);
\path[trans] (n2111) edge[swap] node[tr] {$1|2$} (n2100);
\path[trans] (n2111) edge node[tr,swap,pos=.2,swap] {$1|3$} (n2103);
\path[trans] (n2113) edge node[tr,swap] {$1|1$} (n2133);
\path[trans] (n2130) edge node[tr] {$1|1$} (n2113);
\path[trans] (n2131) edge node[tr] {$1|2$} (n2110);
\path[trans] (n2131) edge node[tr] {$1|3$} (n2113);
\path[trans] (n2133) edge[bend right=10,swap] node[tr] {$1|2$} (n2111);
\path[trans] (n2030) edge[pos=.3] node[tr,swap] {$0|0$} (n2300);
\path[trans] (n2033) edge[swap] node[tr] {$0|1$} (n2300);
\path[trans] (n2030) edge node[tr] {$0|0$} (n2310);
\path[trans] (n2033) edge[swap] node[tr] {$0|1$} (n2310);
\path[trans] (n2033) edge node[tr] {$0|1$} (n2311);
\path[trans] (n2100) edge node[tr,pos=.2,swap] {$0|0$} (n2030);
\path[trans] (n2103) edge[swap] node[tr] {$0|1$} (n2030);
\path[trans] (n2113) edge node[tr] {$0|1$} (n2033);
\path[trans] (n2130) edge[pos=.3] node[tr] {$0|0$} (n2330);
\path[trans] (n2133) edge[pos=.75] node[tr] {$0|1$} (n2330);
\path[trans] (n2133) edge node[tr] {$0|1$} (n2331);
\path[trans] (n2300) edge node[tr] {$0|0$} (n2130);
\path[trans] (n2310) edge node[tr] {$0|1$} (n2103);
\path[trans] (n2311) edge node[tr,pos=.3] {$0|2$} (n2100);
\path[trans] (n2311) edge node[tr] {$0|3$} (n2103);
\path[trans] (n2330) edge node[tr] {$0|1$} (n2113);
\path[trans] (n2331) edge node[tr] {$0|2$} (n2110);
\path[trans] (n2331) edge node[tr] {$0|3$} (n2113);
 \end{tikzpicture}
\end{minipage}
\caption{$\ws_A$, the union of $\sc{\ws_{10000}}$ (left) and  $\sc{\ws_{1000}}$ (right).}
\label{fig:wsa}
\vspace{.5cm}
\end{subfigure}

\begin{subfigure}[b]{.999\linewidth}
\begin{minipage}{.44\linewidth}
\center
\begin{tikzpicture}[auto,scale=0.8]
\tikzstyle{every state}=[shape=ellipse,inner sep=0.5mm,font=\footnotesize,scale=0.6]
\tikzstyle{trans}=[->,thin]
\tikzstyle{tr}=[inner sep=0.2mm,font=\footnotesize,scale=0.6]
\node[state] (n20330) at (5.824000,4.928000) {20330};
\node[state] (n21030) at (9.408000,4.928000) {21030};
\node[state] (n21033) at (7.182000,3.584000) {21033};
\node[state] (n21103) at (9.408000,2.240000) {21103};
\node[state] (n21113) at (5.824000,3.584000) {21113};
\node[state] (n21130) at (3.584000,4.928000) {21130};
\node[state] (n21300) at (8.512000,4.032000) {21300};
\node[state] (n21310) at (8.512000,3.136000) {21310};
\node[state] (n21330) at (4.480000,3.584000) {21330};
\node[state] (n23100) at (7.168000,4.928000) {23100};
\node[state] (n23300) at (3.584000,2.912000) {23300};
\node[state] (n23310) at (4.256000,2.240000) {23310};
\path[trans] (n21030) edge node[tr,swap] {$1|0$} (n21300);
\path[trans] (n21033) edge node[tr,swap] {$1|1$} (n21300);
\path[trans] (n21030) edge[bend left=5] node[tr] {$1|0$} (n21310);
\path[trans] (n21033) edge node[tr,swap] {$1|1$} (n21310);
\path[trans] (n21103) edge node[tr,swap] {$1|1$} (n21030);
\path[trans] (n21113) edge node[tr] {$1|1$} (n21033);
\path[trans] (n21130) edge node[tr] {$1|0$} (n21330);
\path[trans] (n21300) edge[bend left=5] node[tr] {$1|0$} (n21130);
\path[trans] (n21310) edge node[tr] {$1|1$} (n21103);
\path[trans] (n21330) edge node[tr] {$1|1$} (n21113);
\path[trans] (n20330) edge node[tr] {$0|0$} (n23100);
\path[trans] (n21130) edge node[tr] {$0|0$} (n20330);
\path[trans] (n21330) edge node[tr] {$0|0$} (n23300);
\path[trans] (n21330) edge node[tr] {$0|0$} (n23310);
\path[trans] (n23100) edge node[tr] {$0|0$} (n21030);
\path[trans] (n23300) edge node[tr] {$0|0$} (n21130);
\path[trans] (n23310) edge node[tr] {$0|1$} (n21103);
 \end{tikzpicture}
\end{minipage}
\begin{minipage}{.55\linewidth}
\center 
\begin{tikzpicture}[auto,scale=0.6]
\tikzstyle{every state}=[shape=ellipse,inner sep=0.5mm,font=\footnotesize,scale=0.6]
\tikzstyle{trans}=[->,thin]
\tikzstyle{tr}=[inner sep=0.2mm,font=\footnotesize,scale=0.6]
\node[state] (n2030) at (9.632000,4.032000) {2030};
\node[state] (n2033) at (5.824000,4.032000) {2033};
\node[state] (n2103) at (9.632000,6.720000) {2103};
\node[state] (n2113) at (4.032000,4.032000) {2113};
\node[state] (n2130) at (5.824000,6.720000) {2130};
\node[state] (n2133) at (2.240000,4.032000) {2133};
\node[state] (n2300) at (7.392000,4.928000) {2300};
\node[state] (n2310) at (7.840000,3.584000) {2310};
\node[state] (n2330) at (2.240000,6.720000) {2330};
\path[trans] (n2103) edge node[tr] {$1|1$} (n2130);
\path[trans] (n2113) edge node[tr] {$1|1$} (n2133);
\path[trans] (n2130) edge node[tr] {$1|1$} (n2113);
\path[trans] (n2030) edge[pos=.3] node[tr] {$0|0$} (n2300);
\path[trans] (n2033) edge node[tr] {$0|1$} (n2300);
\path[trans] (n2030) edge node[tr] {$0|0$} (n2310);
\path[trans] (n2033) edge node[tr] {$0|1$} (n2310);
\path[trans] (n2103) edge node[tr] {$0|1$} (n2030);
\path[trans] (n2113) edge node[tr] {$0|1$} (n2033);
\path[trans] (n2130) edge node[tr] {$0|0$} (n2330);
\path[trans] (n2133) edge node[tr] {$0|1$} (n2330);
\path[trans] (n2300) edge node[tr] {$0|0$} (n2130);
\path[trans] (n2310) edge node[tr] {$0|1$} (n2103);
\path[trans] (n2330) edge node[tr] {$0|1$} (n2113);
 \end{tikzpicture}
\end{minipage}
\caption{$\ws_B$ corresponds to $\ws_A$ when unused transitions are deleted.}
\label{fig:wsb}
\vspace{.5cm}
\end{subfigure}

\begin{subfigure}[b]{.999\linewidth}
\begin{minipage}{.44\linewidth}
\center
\begin{tikzpicture}[auto,scale=0.8]
\tikzstyle{every state}=[shape=ellipse,scale=0.7]
\tikzstyle{trans}=[->,line width=.5]
\tikzstyle{tr}=[inner sep=0.5mm,scale=0.7]
\node[state] (na) at (9.408000,4.928000) {a};
\node[state] (nb) at (8.512000,3.584000) {b};
\node[state] (nc) at (3.584000,4.928000) {c};
\node[state] (nd) at (4.480000,3.584000) {d};
\node[state] (ne) at (5.824000,3.584000) {e};
\node[state] (nf) at (7.182000,3.584000) {f};
\node[state] (ng) at (9.408000,2.240000) {g};
\node[state] (nh) at (3.584000,2.240000) {h};
\node[state] (ni) at (5.824000,4.928000) {i};
\node[state] (nj) at (7.168000,4.928000) {j};
\path[trans] (na) edge node[tr,swap] {$1|0$} (nb);
\path[trans] (nf) edge node[tr,swap] {$1|1$} (nb);
\path[trans] (ng) edge node[tr,swap] {$1|1$} (na);
\path[trans] (ne) edge node[tr,swap] {$1|1$} (nf);
\path[trans] (nc) edge node[tr,pos=.7,swap] {$1|0$} (nd);
\path[trans] (nb) edge[out=150,in=-20] node[tr,swap,pos=.15] {$1|0$} (nc);
\path[trans] (nb) edge node[tr,swap] {$1|1$} (ng);
\path[trans] (nd) edge node[tr,swap] {$1|1$} (ne);
\path[trans] (ni) edge node[tr] {$0|0$} (nj);
\path[trans] (nc) edge node[tr] {$0|0$} (ni);
\path[trans] (nd) edge node[tr] {$0|0$} (nh);
\path[trans] (nj) edge node[tr] {$0|0$} (na);
\path[trans] (nh) edge node[tr] {$0|0$} (nc);
\path[trans] (nh) edge node[tr,swap] {$0|1$} (ng);
 \end{tikzpicture}
\end{minipage}
\begin{minipage}{.55\linewidth}
\center 
\begin{tikzpicture}[auto,scale=0.8]
\tikzstyle{every state}=[shape=ellipse,scale=0.7]
\tikzstyle{trans}=[->,line width=.5]
\tikzstyle{tr}=[inner sep=0.5mm,scale=0.7]
\node[state] (nK) at (2.240000,4.032000) {K};
\node[state] (nL) at (2.240000,6.720000) {L};
\node[state] (nM) at (4.032000,4.032000) {M};
\node[state] (nN) at (5.824000,4.032000) {N};
\node[state] (nO) at (7.616000,4.032000) {O};
\node[state] (nP) at (9.408000,6.720000) {P};
\node[state] (nQ) at (9.408000,4.032000) {Q};
\node[state] (nR) at (5.824000,6.720000) {R};
\path[trans] (nP) edge node[tr] {$1|1$} (nR);
\path[trans] (nM) edge node[tr] {$1|1$} (nK);
\path[trans] (nR) edge node[tr] {$1|1$} (nM);
\path[trans] (nQ) edge node[tr] {$0|0$} (nO);
\path[trans] (nN) edge node[tr] {$0|1$} (nO);
\path[trans] (nP) edge node[tr] {$0|1$} (nQ);
\path[trans] (nM) edge node[tr] {$0|1$} (nN);
\path[trans] (nR) edge node[tr] {$0|0$} (nL);
\path[trans] (nK) edge node[tr] {$0|1$} (nL);
\path[trans] (nO) edge node[tr] {$0|0$} (nR);
\path[trans] (nO) edge node[tr] {$0|1$} (nP);
\path[trans] (nL) edge node[tr] {$0|1$} (nM);
 \end{tikzpicture}
\end{minipage}
\caption{$\ws_C$ is the simplification of $\ws_B$ by bisimulation.} \label{fig:wsc}
\vspace{.5cm}
\end{subfigure}

\begin{subfigure}[b]{.999\linewidth}
\begin{minipage}{.44\linewidth}
\center
\begin{tikzpicture}[auto,scale=0.8]
\tikzstyle{every state}=[shape=ellipse,scale=0.7]
\tikzstyle{trans}=[->,line width=.5]
\tikzstyle{tr}=[inner sep=0.5mm,scale=0.7]
\node[state] (na) at (9.408000,4.928000) {a};
\node[state] (nb) at (8.512000,3.584000) {b};
\node[state] (nc) at (3.584000,4.928000) {c};
\node[state] (nd) at (4.480000,3.584000) {d};
\node[state] (ne) at (5.824000,3.584000) {e};
\node[state] (nf) at (7.182000,3.584000) {f};
\node[state] (ng) at (9.408000,2.240000) {g};
\node[state] (nh) at (3.584000,2.240000) {h};
\node[state] (ni) at (5.824000,4.928000) {i};
\node[state] (nj) at (7.168000,4.928000) {j};
\path[trans] (na) edge node[tr,swap] {$1|0$} (nb);
\path[trans] (nf) edge node[tr,swap] {$1|1$} (nb);
\path[trans] (ng) edge node[tr,swap] {$1|1$} (na);
\path[trans] (ne) edge node[tr,swap] {$1|1$} (nf);
\path[trans] (nc) edge node[tr,pos=.7,swap] {$1|0$} (nd);
\path[trans] (nb) edge[out=150,in=-20] node[tr,swap,pos=.15] {$1|0$} (nc);
\path[trans] (nb) edge node[tr,swap] {$1|1$} (ng);
\path[trans] (nd) edge node[tr,swap] {$1|1$} (ne);
\path[trans] (ni) edge node[tr] {$0|0$} (nj);
\path[trans] (nc) edge node[tr] {$0|0$} (ni);
\path[trans] (nd) edge node[tr] {$0|0$} (nh);
\path[trans] (nj) edge node[tr] {$0|0$} (na);
\path[trans] (nh) edge node[tr] {$0|0$} (nc);
\path[trans] (nh) edge node[tr,swap] {$0|1$} (ng);
 \end{tikzpicture}
\end{minipage}
\begin{minipage}{.55\linewidth}
\center 
\begin{tikzpicture}[auto,scale=0.8]
\tikzstyle{every state}=[shape=ellipse,scale=0.7]
\tikzstyle{trans}=[->,line width=.5]
\tikzstyle{tr}=[inner sep=0.5mm,scale=0.7,swap]
\node[state] (nK) at (2.240000,6.720000) {K};
\node[state] (nL) at (2.240000,4.032000) {L};
\node[state] (nM) at (4.032000,4.032000) {M};
\node[state] (nM') at (4.032000,6.720000) {M'};
\node[state] (nN) at (5.824000,4.032000) {N};
\node[state] (nO) at (7.616000,4.032000) {O};
\node[state] (nO') at (7.616000,6.720000) {O'};
\node[state] (nP) at (9.408000,4.032000) {P};
\node[state] (nQ) at (9.408000,6.720000) {Q};
\node[state] (nR) at (5.824000,6.720000) {R};
\path[trans] (nP) edge node[tr,pos=.4] {$1|1$} (nR);
\path[trans] (nM') edge node[tr] {$1|1$} (nK);
\path[trans] (nR) edge node[tr,pos=.5,swap] {$1|1$} (nM);
\path[trans] (nQ) edge node[tr] {$0|0$} (nO');
\path[trans] (nN) edge node[tr] {$0|1$} (nO);
\path[trans] (nP) edge node[tr] {$0|1$} (nQ);
\path[trans] (nM) edge node[tr] {$0|1$} (nN);
\path[trans] (nR) edge node[tr,pos=.6] {$0|0$} (nL);
\path[trans] (nK) edge node[tr] {$0|1$} (nL);
\path[trans] (nO) edge node[tr,pos=.5,swap] {$0|0$} (nR);
\path[trans] (nO) edge node[tr] {$0|1$} (nP);
\path[trans] (nL) edge node[tr] {$0|1$} (nM);
\path[trans] (nO') edge node[tr] {$0|0$} (nR);
\path[trans] (nR) edge node[tr] {$1|1$} (nM');
 \end{tikzpicture}
\end{minipage}
\caption{$\ws_D$ is the simplification of $\ws_C$, using the fact that
  the successions of symbols $101$ and $010$ cannot appear. 
  The transducers to the left and to the right are called
   $\ws_\texttt{a}$ and $\ws_\texttt{b}$, respectively.}
\label{fig:wsd}
\vspace{.5cm}
\end{subfigure}
\caption{The different steps of simplification of $\ws_A$.}
\end{figure}

We see that the transducer $\ws_A$ never reads $2$, $3$, nor $4$. 
Thus the transitions that write $2$, $3$, or $4$ are never used in a tiling by $\ws$.
Let $\ws_B$  (see Figure~\ref{fig:wsb}) be the transducer $\ws_A$ after removing these unused
transitions, and deleting states that cannot appear in a tiling of a row (i.e., sources and sinks).
Then $t$ is a tiling by $\ws_A$ if and only if $t$ is a tiling by $\ws_B$, and $\ws_B$ is aperiodic if and only if $\ws_A$ is.

Now we use bisimulation to slightly simplify the transducer $\ws_B$.
The states 23300 and 23310 have the same incoming transitions, and
can therefore be coalesced into one state.
The same goes for states 21300 and 21310, and for states 2300 and 2310.
Once we coalesce all those states, we obtain the Wang set $\ws_C$
depicted in Figure~\ref{fig:wsc}.

$\ws_B$ and $\ws_C$ are equivalent.
Therefore, $\ws_B$ is aperiodic if and only if $\ws_C$ is aperiodic.

\begin{proposition}\label{prop:010}
Let $W$ be the set of biinfinite words which do not contain the words 
$010$ and $101$ as factors.
Let $u,v,w$ s.t. $u \ws_C v \ws_C w$. Then $v \in W$.

  In particular, let $(w_i)_{i\in\mathbb{Z}}$ be a biinfinite
  sequence of biinfinite binary words such that $w_i\ws_C w_{i+1}$
  for every $i\in \mathbb{Z}$. Then, for every $i\in \mathbb{Z}$, $w_i
  \in W$.
\end{proposition}
\begin{proof}
A quick inspection shows that the transducer $\ws_C$ does not accept
as input a word which contains $101$, and it does not produce a word which
contains $010$ as an output.
\end{proof}

In a tiling by $\ws_C$, the transition from Q to O is never followed
by a transition from O to P, otherwise it writes a $101$. Similarly, a transition from M to K is never preceded by a transition
from L to M, otherwise it reads a $010$. Thus, there is a bijection
between tilings by $\ws_C$ and tilings by $\ws_D$ (Figure~\ref{fig:wsd}).

We therefore have:
\begin{proposition}
$\ws$ is aperiodic if and only if $\ws_D$ is aperiodic.  
\end{proposition}  

\subsection{From \texorpdfstring{$\ws_D$ to $T_0, T_1, T_2$}{TD to T0,T1,T2}}
\label{ssec:seq}
Let $\ws_\texttt{a}$ and $\ws_\texttt{b}$ be the two connected components of $\ws_D$.
For a word $w\in \{\texttt{a},\texttt{b}\}^*$, let $\ws_{w} = \ws_{w[1]} \circ \ws_{w[2]} \circ\ldots \circ\ws_{w[\vert w\vert]}$.
The following fact can be easily checked by computer or by hand:
\begin{fact}
The transducers $\sc{\ws_{\texttt{bb}}}$, $\sc{\ws_{\texttt{aaa}}}$
and $\sc{\ws_{\texttt{babab}}}$ are empty.
\end{fact}

This implies that if $t$ is a tiling by $\ws_C$, then there
exists a biinfinite binary word $w\in \{\texttt{b},\texttt{aa},\texttt{bab}\}^{\mathbb{Z}}$ such
that $t(x,y)\in T(\ws_{w[y]})$ for every $y\in\mathbb{Z}$. That is,
$t$ is an image of a tiling by $\ws_\texttt{b} \cup \ws_\texttt{aa} \cup \ws_{\texttt{bab}}$. 

We will now simplify the three transducers.
\paragraph{Case of $\ws_\texttt{b}$.}

In $\ws_\texttt{b}$, every path eventually goes to the state ``N''. Thus
$\ws_{b}$ is equivalent to the following transducer (written in a compact form):

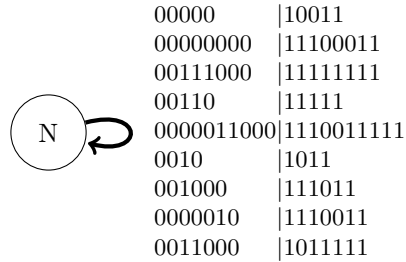
\begin{center}
\begin{tikzpicture}[->,node distance=5cm,auto,initial text=]
        
  \tikzstyle{every state}=[minimum size=10mm]
  
  \node[state] (A) at (0,0)                   {N};
  \path[ultra thick] (A) edge[loop right]   node[right] {
    $ \small
 \begin{array}{l@{{}|{}}l}
      00000&10011\\
00000000&11100011\\
00111000&11111111\\
00110&11111\\
0000011000&1110011111\\
0010&1011\\
001000&111011\\
0000010&1110011\\
0011000&1011111\\
     \end{array}$
  } (A);
\end{tikzpicture}
\end{center}

In the previous transducer, the last 4 transitions are never used in a tiling of the plane, since they read 010 or write 101. 
This allows us to simplify the transducer into:

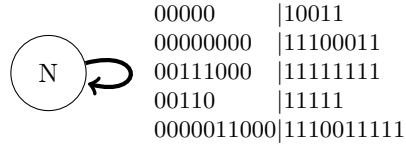
\begin{center}
\begin{tikzpicture}[->,node distance=5cm,auto,initial text=]
        
  \tikzstyle{every state}=[minimum size=10mm]
  
  \node[state] (A) at (0,0)                   {N};
  \path[ultra thick] (A) edge[loop right]   node[right] {
    $ \small
 \begin{array}{l@{{}|{}}l}
      00000&10011\\
00000000&11100011\\
00111000&11111111\\
00110&11111\\
0000011000&1110011111\\
     \end{array}$
  } (A);
\end{tikzpicture}
\end{center}

This transducer is equivalent to $T_{0}$, that recalled here for
comparison:

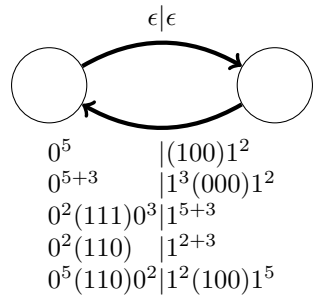
\begin{center}
\begin{tikzpicture}[->,node distance=5cm,auto,initial text=,scale=0.3]
        
        \tikzstyle{every state}=[minimum size=10mm]

\node[state] (A) at (0,0)                   {};
\node[state] (F) at (10,0) {};

\path[ultra thick] (A) edge[bend left]   node[above] {$\epsilon|\epsilon$} (F);
\path[ultra thick] (F) edge[bend left]   node[below] {
$  \begin{array}{l@{|}l}
0^{5}&(100)1^{2}\\
0^{5+3}&1^{3}(000)1^{2}\\
0^{2}(111)0^{3}&1^{5+3}\\
0^{2}(110)&1^{2+3}\\
0^{5}(110)0^{2}&1^{2}(100)1^{5}\\
\end{array}$} (A);
\end{tikzpicture}
\end{center}

\begin{figure}
\begin{subfigure}[b]{.51\linewidth}
\center
\begin{tikzpicture}[auto,scale=0.85]
\tikzstyle{every state}=[shape=ellipse,inner sep=0.5mm,font=\footnotesize,scale=0.6]
\tikzstyle{trans}=[->,thin]
\tikzstyle{tr}=[inner sep=0.2mm,font=\footnotesize,scale=0.6]
\node[state] (nba) at (4.5,2) {ba};
\node[state] (nch) at (7,0) {ch};
\node[state] (ncj) at (2.5,2) {cj};
\node[state] (ndc) at (0,0) {dc};
\node[state] (neb) at (3.5,0.8) {eb};
\node[state] (ngb) at (3.5,3.2) {gb};
\path[trans] (ngb) edge node[tr,swap] {$111|000$} (ncj);
\path[trans] (neb) edge node[tr,pos=0.4] {$111|000$} (nch);
\path[trans] (neb) edge node[tr,swap] {$11|11$} (nba);
\path[trans] (ncj) edge node[tr,swap] {$11|00$} (neb);
\path[trans] (nch) edge node[tr] {$1|0$} (ndc);
\path[trans] (nch) edge[bend right=45] node[tr,swap,pos=0.8] {$11111100|11000000$} (ngb);
\path[trans] (nba) edge node[tr,swap] {$1|0$} (ngb);
\path[trans] (ndc) edge[bend left=45] node[tr,pos=0.8] {$1111|0111$} (ngb);
\path[trans] (nch) edge node[tr,swap,pos=0.6] {$0001|0000$} (nba);
\path[trans] (ndc) edge node[tr] {$00|00$} (ncj);
 \end{tikzpicture}
\subcaption{$\sc{\ws_\texttt{aa}}$}
\label{fig:wsaa}
\end{subfigure}
\hspace{.3cm}
\begin{subfigure}[b]{.45\linewidth}
\center
\begin{tikzpicture}[auto,scale=0.85]
\tikzstyle{every state}=[shape=ellipse,inner sep=0.5mm,font=\footnotesize,scale=0.6]
\tikzstyle{trans}=[->,thin]
\tikzstyle{tr}=[inner sep=0.2mm,font=\footnotesize,scale=0.6]
\node[state] (nLaO) at (5.8,6) {LaO};
\node[state] (nMbK) at (2,8) {MbK};
\node[state] (nMbR) at (7.,5) {MbR};
\node[state] (nNeR) at (3.2,6) {NeR};
\node[state] (nNcL) at (2,5) {NcL};
\node[state] (nPbN) at (5.8,7) {PbN};
\node[state] (nQcO) at (7.,8) {QcO};
\node[state] (nRcO) at (3.2,7) {RcO};
\path[trans] (nPbN) edge node[tr] {$1|1$} (nRcO);
\path[trans] (nRcO) edge node[tr] {$10|11$} (nNeR);
\path[trans] (nQcO) edge node[tr,swap] {$000000011|001111111$} (nMbK);
\path[trans] (nQcO) edge node[tr,swap] {$0000|1100$} (nMbR);
\path[trans] (nNcL) edge node[tr,swap] {$00000|11111$} (nNeR);
\path[trans] (nPbN) edge node[tr] {$0|1$} (nQcO);
\path[trans] (nMbR) edge node[tr] {$0|0$} (nNcL);
\path[trans] (nMbK) edge node[tr] {$0|1$} (nNcL);
\path[trans] (nNeR) edge node[tr] {$000000|111111$} (nLaO);
\path[trans] (nLaO) edge node[tr] {$0|0$} (nMbR);
\path[trans] (nLaO) edge node[tr] {$0000|1111$} (nPbN);
\path[trans] (nRcO) edge node[tr,swap] {$1100|1111$} (nMbK);
 \end{tikzpicture}
\subcaption{$\sc{\ws_{\texttt{bab}}}$}
\label{fig:wsbab}
\end{subfigure}
\caption{$\sc{\ws_\texttt{aa}}$ and $\sc{\ws_{\texttt{bab}}}$.}
\end{figure}

\paragraph{Case of $\ws_\texttt{aa}$.}
The transducer $\sc{\ws_\texttt{aa}}$ is depicted in Figure~\ref{fig:wsaa} in a compact form.
In this transducer, every path eventually goes to the state ``eb''. Then $\sc{\ws_\texttt{aa}}$ is equivalent to the following transducer (written in a compact form):

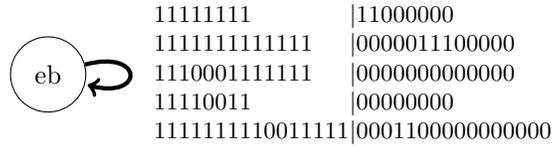
\begin{center}
\begin{tikzpicture}[->,node distance=5cm,auto,initial text=]
        
  \tikzstyle{every state}=[minimum size=10mm]
  
  \node[state] (A) at (0,0)                   {eb};
  \path[ultra thick] (A) edge[loop right]   node[right] {
    $ \small
 \begin{array}{l@{{}|{}}l}
      11111111&11000000\\
1111111111111&0000011100000\\
1110001111111&0000000000000\\
11110011&00000000\\
1111111110011111&0001100000000000\\
     \end{array}$
  } (A);
\end{tikzpicture}
\end{center}

This transducer is clearly equivalent to $T_1$, recalled here for
convenience:

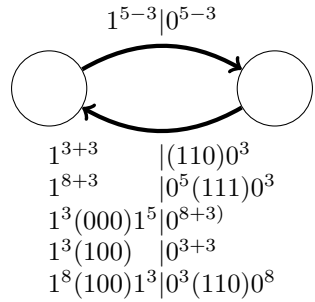
\begin{center}
\begin{tikzpicture}[->,node distance=5cm,auto,initial text=,scale=0.3]
        
        \tikzstyle{every state}=[minimum size=10mm]

\node[state] (A) at (0,0)                   {};
\node[state] (F) at (10,0){};
\path[ultra thick] (A) edge[bend left]   node[above] {$1^{5-3}|0^{5-3}$} (F);
\path[ultra thick] (F) edge[bend left]   node[below] {
$  \begin{array}{l@{|}l}
 1^{3+3}&(110)0^{3}\\
1^{8+3}&0^{5}(111)0^{3}\\
1^{3}(000)1^{5}&0^{8+3)}\\
1^{3}(100)&0^{3+3}\\
1^{8}(100)1^{3}&0^{3}(110)0^{8}
\end{array}$} (A);
	
\end{tikzpicture}
\end{center}

\paragraph{Case of $\ws_{\texttt{bab}}$.}
The transducer $\sc{\ws_{\texttt{bab}}}$ is depicted in Figure~\ref{fig:wsbab}.

In this transducer, every path eventually goes to the state ``NeR''. Then $\sc{\ws_{\texttt{bab}}}$ is equivalent to the following transducer (written in compact form):

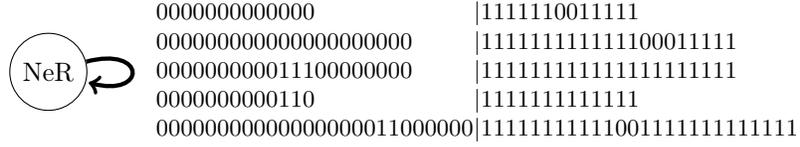
\begin{center}
\begin{tikzpicture}[->,node distance=5cm,auto,initial text=]
        
  \tikzstyle{every state}=[minimum size=10mm]
  
  \node[state] (A) at (0,0)                   {NeR};
  \path[ultra thick] (A) edge[loop right]   node[right] {
    $      \small
    \begin{array}{l@{{}|{}}l}
      0000000000000&1111110011111\\
000000000000000000000&111111111111100011111\\
000000000011100000000&111111111111111111111\\
0000000000110&1111111111111\\
00000000000000000011000000&11111111111001111111111111\\
     \end{array}$
  } (A);
\end{tikzpicture}
\end{center}

This transducer is clearly equivalent to $T_2$, which we recall here for the
reader's convenience:

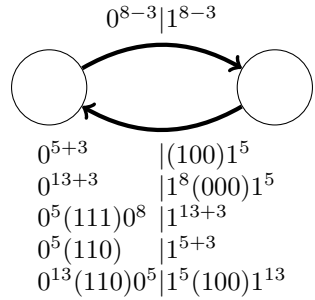
\begin{center}
\begin{tikzpicture}[->,node distance=5cm,auto,initial text=,scale=0.3]
        
        \tikzstyle{every state}=[minimum size=10mm]

\node[state] (A) at (0,0)                   {};
\node[state] (F) at (10,0) {};

\path[ultra thick] (A) edge[bend left]   node[above] {$0^{8-3}|1^{8-3}$} (F);
\path[ultra thick] (F) edge[bend left]   node[below] {
$  \begin{array}{l@{|}l}
0^{5+3}&(100)1^{5}\\
0^{13+3}&1^{8}(000)1^{5}\\
0^{5}(111)0^{8}&1^{13+3}\\
0^{5}(110)&1^{5+3}\\
0^{13}(110)0^{5}&1^{5}(100)1^{13}\\
\end{array}$} (A);
\end{tikzpicture}
\end{center}

\section{From \texorpdfstring{$T_n,T_{n+1},T_{n+2}$ to $T_{n+1},T_{n+2}, T_{n+3}$}{Tn,Tn+1,Tn+2 to Tn+1,Tn+2,Tn+3}}
\label{sec:recur}

In this section, we prove:
\begin{theorem}
   For all words $u,v$ we have  $u T_{n+3} v \iff u T_{n+1} T_n T_{n+1} v$.
\end{theorem}  
For the reader's convenience, we recall the definition of the family of
transducers, and we introduce notations for the transitions.

$T_n$ for $n$ even:
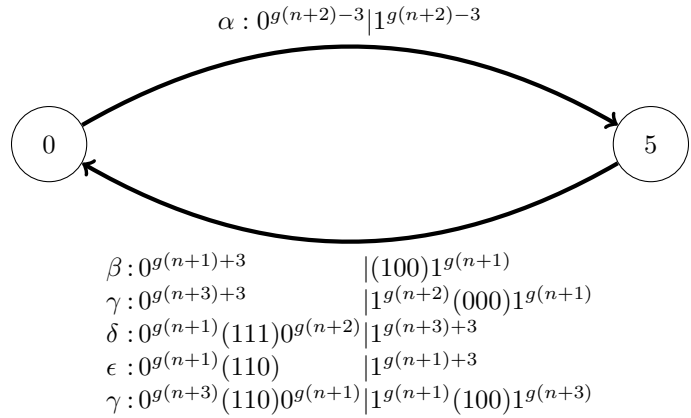
\begin{center}
\begin{tikzpicture}[->,node distance=5cm,auto,initial text=]
        
        \tikzstyle{every state}=[minimum size=10mm]

\node[state] (A) at (0,0)                   {$0$};
\node[state] (F) at (8,0) {$5$};

\path[ultra thick] (A) edge[bend left]   node[above] {$\alpha:   0^{g(n+2)-3}|1^{g(n+2)-3}$} (F);
\path[ultra thick] (F) edge[bend left]   node[below] {
$  \begin{array}{l@{\,:\,}l@{|}l}
\beta& 0^{g(n+1)+3}&(100)1^{g(n+1)}\\
\gamma& 0^{g(n+3)+3}&1^{g(n+2)}(000)1^{g(n+1)}\\
\delta& 0^{g(n+1)}(111)0^{g(n+2)}&1^{g(n+3)+3}\\
\epsilon& 0^{g(n+1)}(110)&1^{g(n+1)+3}\\
\omega& 0^{g(n+3)}(110)0^{g(n+1)}&1^{g(n+1)}(100)1^{g(n+3)}\\
\end{array}$} (A);
\end{tikzpicture}
\end{center}

$T_{n+1}$ for $n$ even:
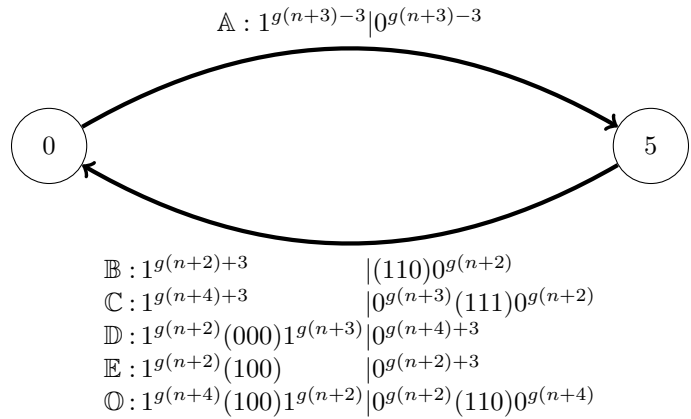
\begin{center}
\begin{tikzpicture}[->,node distance=5cm,auto,initial text=]
        
        \tikzstyle{every state}=[minimum size=10mm]

\node[state] (A) at (0,0)                   {$0$};
\node[state] (F) at (8,0) {$5$};
\path[ultra thick] (A) edge[bend left]   node[above] {$\mathbb{A}:  1^{g(n+3)-3}|0^{g(n+3)-3}$} (F);
\path[ultra thick] (F) edge[bend left]   node[below] {
$  \begin{array}{l@{\,:\,}l@{|}l}
\mathbb{B}& 1^{g(n+2)+3}&(110)0^{g(n+2)}\\
\mathbb{C}& 1^{g(n+4)+3}&0^{g(n+3)}(111)0^{g(n+2)}\\
\mathbb{D}& 1^{g(n+2)}(000)1^{g(n+3)}&0^{g(n+4)+3}\\
\mathbb{E}& 1^{g(n+2)}(100)&0^{g(n+2)+3}\\
\mathbb{O}& 1^{g(n+4)}(100)1^{g(n+2)}&0^{g(n+2)}(110)0^{g(n+4)}
\end{array}$} (A);
	
\end{tikzpicture}
\end{center}

Before going through the proof, some remarks:

\begin{itemize}
	\item
	  $T_n$ for $n$ even and $n$ odd are essentially similar. This means it is sufficient to
	  prove the result for $n$ even.
	\item Apply the following transformations to $T_n$: exchange input
	  and output, reverse the direction of the edge, reverse (i.e., take the mirror) 
	  the words, and exchange symbols $0$ and $1$. 
	  Then we obtain $T_n$ again (for $n$ even, with $\beta$ playing the role of
	  $\epsilon$, $\delta$ the role of $\gamma$, and $\alpha$ and
	  $\omega$ their own role). This internal symmetry will be used heavily in the proofs.
	\item All transitions are symmetric and easy to understand, except the
	  self-symmetric tiles $\omega$ and $\mathbb{O}$. These
	  transitions actually cannot occur in the tiling of the plane,
	  but a transition of shape $\omega$ or $\mathbb{O}$ large enough can appear in a
	  large enough finite strip. 
	  Therefore, it is not possible to accomplish the proof without speaking
	  about these transitions, even though they cannot appear in a tiling
	  of the plane.
  \end{itemize}

  We now proceed to prove the result.
As said before, we suppose that $n$ is even, and we will look at
the sequence of transducers $T_{n+1} \circ T_n \circ T_{n+1}$.

Note that the output of $T_n$ consists essentially of long sequences
of the symbol $1$, and a few occurrences of $100$ and
$000$ interspersed. We call these two words ``markers''.
Because the output of $T_n$ should be fed to $T_{n+1}$, the distance
between the markers that $T_n$ produces should be within what
$T_{n+1}$ can read.

\newmuskip{\medmuskipsave}
\medmuskipsave=\medmuskip
\setlength{\medmuskip}{0mu}

The following table represents the possible distance between two
consecutive markers (i.e., 000 and 100) as inputs of $T_{n+1}$.
\[\begin{array}{l|l|ll}
\text{First Marker} & \text{Second Marker} & \text{Distance}\\
\hline
(000) \text{ from } \mathbb{D} & (000) \text{ from } \mathbb{D} & g(n+5) & \rdelim\}{9}{3.5cm}[$\begin{matrix}+ag(n+4) + bg(n+5)\\ a,b \in \mathbb{N} \end{matrix}$]\\
(000) \text{ from } \mathbb{D} & (100) \text{ from } \mathbb{E} & g(n+5)\\
(000) \text{ from } \mathbb{D} & (100) \text{ from } \mathbb{O} & g(n+5) + g(n+3)\\
(100) \text{ from } \mathbb{E} & (000) \text{ from } \mathbb{D} & g(n+4) \\
(100) \text{ from } \mathbb{E} & (100) \text{ from } \mathbb{E} & g(n+4) \\
(100) \text{ from } \mathbb{E} & (100) \text{ from } \mathbb{O} & g(n+5) \\
(100) \text{ from } \mathbb{O} & (000) \text{ from } \mathbb{D} & g(n+4) + g(n+2) \\
(100) \text{ from } \mathbb{O} & (100) \text{ from } \mathbb{E} & g(n+4) + g(n+2)\\
(100) \text{ from } \mathbb{O} & (100) \text{ from } \mathbb{O} & 2g(n+4) \\
\hline
\end{array}\]

For example, between the marker $000$ from $\mathbb{D}$ and $100$ from $\mathbb{O}$, one has $3$ (the size of the first marker) plus $g(n+3)$ (the letters in $\mathbb{D}$ after the marker), plus $g(n+3)-3$ (the letters in $\mathbb{A}$) plus $g(n+4)$ (the letters of $\mathbb{O}$ before the marker). That is $3+g(n+3)+g(n+3)-3+g(n+4)=g(n+3)+g(n+5)$.

\setlength{\medmuskip}{\medmuskipsave}

We mean by distance the absolute value between the positions of the first letter of each marker.
To prove the main result, we will prove that the transitions in the
transducer $T_n$ (when surrounded by transducers $T_{n+1}$) must be done in a certain order.

In the following, we deliberately omit the transition $\alpha$:
when we say that $\gamma\beta$ cannot appear, we mean that it is
impossible to see the transitions $\gamma$, followed by $\alpha$ and then
$\beta$ in a run of the transducer $T_n$ (when surrounded by
transducers $T_{n+1}$).
\begin{lemme}
The following words cannot appear:
\begin{itemize}
\item
  $\gamma\omega$,$\gamma\gamma$,$\gamma\beta,\beta\omega,\beta\beta,\beta\epsilon\beta,\gamma\epsilon\beta$,
$  \beta\delta\epsilon\beta$,
$\gamma\delta\epsilon\beta$
\item
  $\omega\delta,\delta\delta,\epsilon\delta,\omega\epsilon,\epsilon\epsilon,\epsilon\beta\epsilon,\epsilon\beta\delta$,
  $\epsilon\beta\gamma\epsilon$, $\epsilon\beta\gamma\delta$
\end{itemize}  
\end{lemme}
\begin{proof}
All the following successions of transitions are impossible due to the
input constraints on $T_{n+1}$:

\begin{center}	
\begin{tabular}{l|l}
Case & What it would produce (which cannot be fed to $T_{n+1}$) \\
\hline
$\gamma\omega$ & (000) and (100) separated by $g(n+1)+g(n+3)$\\
$\gamma\gamma$ & (000) and (000) separated by $g(n+4)$\\
$\gamma\beta$ & (000) and (100) separated by $g(n+3)$\\
$\beta\omega$ &(100) and (100) separated by $g(n+1)+g(n+3)$\\
$\beta\beta$ &(100) and (100) separated by $g(n+3)$\\
$\beta\epsilon\beta$ &(100) and (100) separated by $2g(n+3)$\\
$\gamma\epsilon\beta$ &(000) and (100) separated by $2g(n+3)$\\
$\beta\delta\epsilon\beta$ & (100) and (100) separated by $2g(n+4) + g(n+1)$\\
$\gamma\delta\epsilon\beta$ & (000) and (100) separated by $2g(n+4) + g(n+1)$\\
\hline
\end{tabular}
\end{center}
All other cases follow by symmetry.
\end{proof}

\begin{lemme}
	$\omega$ cannot appear.
\end{lemme}
\begin{proof}
Case disjunction on what appears before:
\begin{center}	
\begin{tabular}{r|l}
Case & What it would produce (which cannot be fed to $T_{n+1}$)\\
\hline
$\beta\omega$ & see above\\
$\gamma\omega$ & see above \\
$\beta\delta\omega$ & (100) and (100) separated by \\&~~$g(n+4) + g(n+3)+g(n+1)$\\
$\gamma\delta\omega$ & (000) and (100) separated by \\&~~$g(n+4) + g(n+3)+g(n+1)$\\
$\beta\epsilon\omega$ & (100) and (100) separated by \\&~~$g(n+4) + 2g(n+1)$\\
$\gamma\epsilon\omega$ & (000) and (100) separated by \\&~~$g(n+4) + 2g(n+1)$\\
$\beta\delta\epsilon\omega$ & (100), (100) separated by \\&~~$g(n+5) + g(n+3) + g(n+1) = 2g(n+4) +  2g(n+1)$\\
$\gamma\delta\epsilon\omega$ & (000), (100) separated by \\&~~$g(n+5) + g(n+3) + g(n+1) = 2g(n+4) +  2g(n+1)$\\
\hline
\end{tabular}
\end{center}
\end{proof}

\begin{lemme}
	$\mathbb{O}$ cannot appear.
\end{lemme}
\begin{proof}
Suppose that $\mathbb{O}$ appears in the top transducer (i.e., the
transducers with input $T_n$).
This means the $(100)$ marker is generated, the only possibility being by $\beta$.

We prove there is no possibility to find transitions after this $\beta$.

\begin{center}	
\begin{tabular}{r|l}
Case & Why it is impossible to start from $\mathbb{O}$\\
\hline
$\beta\gamma$ & (100) and (000) separated by $g(n+4)$\\
$\beta\delta\beta$ & (100) and (100) separated by $g(n+4)+g(n+3)$\\
$\beta\delta\gamma$ & (100) and (000) separated by $g(n+4) + g(n+1) + g(n+3)$\\
$\beta\delta\epsilon\beta$ & (100) and (100) separated by $2g(n+4) + g(n+1) $\\
$\beta\delta\epsilon\gamma$ & (100) and (000) separated by $2g(n+4) + g(n+3)$\\
$\beta\epsilon\gamma$ & (100) and (000) separated by $g(n+5)$\\
\hline
\end{tabular}
\end{center}	

By symmetry, $\mathbb{O}$ cannot appear in the bottom transducer.
\end{proof}	

Now that $\mathbb{O}$ has disappeared, the possible distances between the
markers are greatly simplified.
\[
	  \begin{array}{l|l|ll}
\text{First Marker} & \text{Second Marker} & \text{Distance}\\
\hline
(000)  & (000) & g(n+5) & \rdelim\}{4}{4cm}[$\begin{matrix}+ag(n+4)+bg(n+5) \\ a,b\in \mathbb{N}\end{matrix}$]\\
(000)  & (100) & g(n+5)\\
(100)  & (000) & g(n+4) \\
(100)  & (100) & g(n+4) \\
\hline
\end{array}
\]

\begin{lemme}
The following words do not appear: $\beta\epsilon$, $\epsilon\beta$
 $\beta\delta\beta$, $\delta\gamma\delta$,
 as well as $\epsilon\gamma\epsilon$ and $\gamma\delta\gamma$
\end{lemme}	
\begin{proof}
$\beta\epsilon$ should be followed by $\gamma$ which leads to $(100)$ and $(000)$ separated by $g(n+5)$.

$\epsilon\beta$ should be preceded by a $\delta$, which cannot be
preceded by anything.

\begin{center}	
\begin{tabular}{r|l}
Case & Why it is impossible\\
\hline
$\beta\delta\beta$ & $(100),(100)$ separated by $g(n+4)+g(n+3)$\\
$\gamma\delta\gamma$ & $(000),(000)$ separated by $g(n+5)+g(n+2)$\\
\hline
\end{tabular}
\end{center}	

The last two follow by symmetry.

\end{proof}

\begin{lemme}
	Every biinfinite path on the transducer $T_n$, when it is
        surrounded by transducers $T_{n+1}$, can be written as
    paths on the following graph:
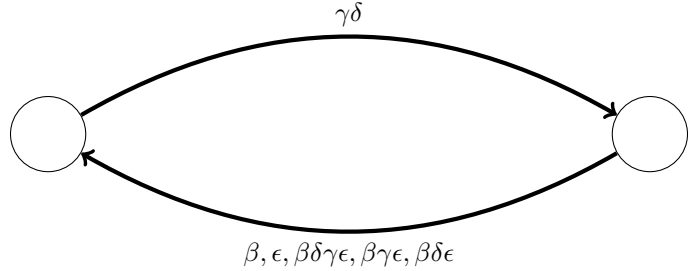
\begin{center}
\begin{tikzpicture}[->,node distance=5cm,auto,initial text=]
        
        \tikzstyle{every state}=[minimum size=10mm]

\node[state] (A) at (0,0)                   {};
\node[state] (F) at (8,0) {};
\path[ultra thick] (A) edge[bend left]   node[above] {$\gamma\delta$} (F);
\path[ultra thick] (F) edge[bend left]   node[below]
{$\beta,\epsilon,\beta\delta\gamma\epsilon,   \beta\gamma\epsilon,\beta\delta\epsilon$} (A);
\end{tikzpicture}
\end{center}	
\end{lemme}
\begin{proof}
Clear: all other words are forbidden by the previous lemmas.
\end{proof}	
Recall that in this picture, words $\alpha$ have been forgotten.
We now rewrite it adding the transitions $\alpha$.
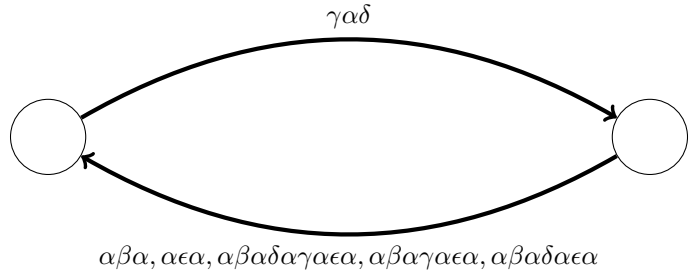
\begin{center}
\begin{tikzpicture}[->,node distance=5cm,auto,initial text=]
        
        \tikzstyle{every state}=[minimum size=10mm]

\node[state] (A) at (0,0)                   {};
\node[state] (F) at (8,0) {};
\path[ultra thick] (A) edge[bend left]   node[above] {$\gamma\alpha\delta$} (F);
\path[ultra thick] (F) edge[bend left]   node[below]
{$\alpha\beta\alpha,\alpha\epsilon\alpha,\alpha\beta\alpha\delta\alpha\gamma\alpha\epsilon\alpha,   \alpha\beta\alpha\gamma\alpha\epsilon\alpha,\alpha\beta\alpha\delta\alpha\epsilon\alpha$} (A);
\end{tikzpicture}
\end{center}	

All transitions in the picture will be called \emph{meta-transitions}.

We now have a more accurate description of the behavior of the
transducer $T_n$ when surrounded by transducers $T_{n+1}$.
This will be sufficient to prove the results.
We will see indeed that each of the six meta-transitions
depicted can be completed in only one way by transitions of $T_{n+1}$.
This will give us six tiles, which (almost) correspond to the
transitions of $T_{n+3}$.

We will use drawings to prove the result. Let's first draw all tiles. 
The bottom corresponds to the input, and the top to the output.
The colors indicate the markers:
the blue (resp. black, red, green) corresponds to $111$ (resp. $110$, $100$ and $000$).

\newcommand\Aalpha[1]{\begin{scope}[xshift=#1cm]
  \draw (0,0) rectangle (2.9,1);
  \node at (1.5,0.5) {$\alpha$};
\end{scope}
}

\newcommand\Abeta[1]{\begin{scope}[xshift=#1cm]
  \draw (0,0) rectangle (2.1,1);
  \draw[fill=red] (0,0.5) rectangle (0.1,1);
  \node at (1,0.5) {$\beta$};
\end{scope}
}

\newcommand\Agamma[1]{\begin{scope}[xshift=#1cm]
  \draw (0,0) rectangle (5.1,1);
  \draw[fill=green] (3,0.5) rectangle (3.1,1);
  \node at (2.5,0.5) {$\gamma$};    
\end{scope}
}

\newcommand\Adelta[1]{\begin{scope}[xshift=#1cm]
  \draw (0,0) rectangle (5.1,1);
  \draw[fill=blue] (2,0) rectangle (2.1,0.5);
  \node at (2.5,0.5) {$\delta$};    
\end{scope}
}

\newcommand\Aepsilon[1]{\begin{scope}[xshift=#1cm]
  \draw (0,0) rectangle (2.1,1);
  \draw[fill=black] (2.1,0) rectangle (2,.5);
  \node at (1,0.5) {$\epsilon$};    
\end{scope}
}

\renewcommand\AA[1]{\begin{scope}[xshift=#1cm]
  \draw (0,0) rectangle (4.9,1);
  \node at (2.5,0.5) {$\mathbb{A}$};
\end{scope}
}

\newcommand\AB[1]{\begin{scope}[xshift=#1cm]
  \draw (0,0) rectangle (3.1,1);
  \draw[fill=black] (0,0.5) rectangle (0.1,1);
  \node at (1.5,0.5) {$\mathbb{B}$};
\end{scope}
}

\newcommand\AC[1]{\begin{scope}[xshift=#1cm]
  \draw (0,0) rectangle (8.1,1);
  \draw[fill=blue] (5,0.5) rectangle (5.1,1);
  \node at (4,0.5) {$\mathbb{C}$};
\end{scope}
}

\newcommand\AD[1]{\begin{scope}[xshift=#1cm]
  \draw (0,0) rectangle (8.1,1);
  \draw[fill=green] (3,0) rectangle (3.1,0.5);
  \node at (4,0.5) {$\mathbb{D}$};    
\end{scope}
}

\renewcommand\AE[1]{\begin{scope}[xshift=#1cm]
  \draw (0,0) rectangle (3.1,1);
  \draw[fill=red] (3,0) rectangle (3.1,.5);
  \node at (1.5,0.5) {$\mathbb{E}$};
\end{scope}
}

First, the transitions of $T_n$, seen as tiles:

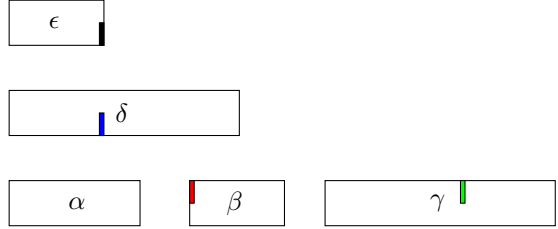
\begin{center}
\begin{tikzpicture}[scale=0.6]
	\Aalpha{0}
	\Abeta{4}
	\Agamma{7}
\begin{scope}[yshift=2cm]
	\Adelta{0}
\end{scope}
\begin{scope}[yshift=4cm]
	\Aepsilon{0}
\end{scope}
\end{tikzpicture}	
\end{center}

Then the transitions of $T_{n+1}$:

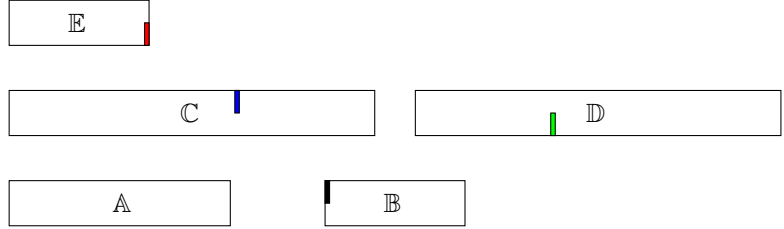
\begin{center}
\begin{tikzpicture}[scale=0.6]
	\AA{0}
	\AB{7}
\begin{scope}[yshift=2cm]
	\AC{0}
	\AD{9}
\end{scope}
\begin{scope}[yshift=4cm]
	\AE{0}
\end{scope}
\end{tikzpicture}	
\end{center}

We now first look at $\gamma\delta$. By necessity, the following
transitions of $T_{n+1}$ should surround it:

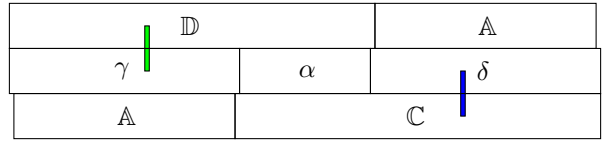
\begin{center}
\begin{tikzpicture}[scale=0.6]
	\Agamma{0}
	\Aalpha{5.1}
	\Adelta{8}
	\begin{scope}[yshift=-1cm]
		  \AA{0.1}
		  \AC{5}
	\end{scope}
	\begin{scope}[yshift=1cm]
		  \AD{0}
		  \AA{8.1}
	\end{scope}	
\end{tikzpicture}	
\end{center}

Note that the three transducers are aligned (up to a shift of $\pm 3$) when $\gamma\alpha\delta$ is
present. 
As all other meta-transitions are enclosed by the meta-transition
$\gamma\alpha\delta$, this means that in an execution
of $T_{n+1} \circ T_n \circ T_{n+1}$, every other meta-transition should
be surrounded above and below by transitions of $T_{n+1}$ that almost
align with it. Moreover, the transitions of $T_{n+1}$ below should
begin by $\mathbb{A}$ and the transitions of $T_{n+1}$ above should
end with $\mathbb{A}$. It turns out that there is only one way to do this for
any of the other meta-transitions.

This gives for $\epsilon$ and $\beta$:
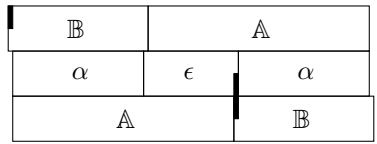
\begin{center}
\begin{tikzpicture}[scale=0.6]
	\Aalpha{0}
	\Aepsilon{2.9}
	\Aalpha{5}
	\begin{scope}[yshift=-1cm]
		  \AA{0}
		  \AB{4.9}
	\end{scope}
	\begin{scope}[yshift=1cm]
		  \AB{-0.1}
		  \AA{3}
	\end{scope}
\end{tikzpicture}
\end{center}

\begin{center}
\begin{tikzpicture}[scale=0.6]
	\Aalpha{0}
	\Abeta{2.9}
	\Aalpha{5}
	\begin{scope}[yshift=-1cm]
		  \AA{0}
		  \AE{4.9}
	\end{scope}
	\begin{scope}[yshift=1cm]
		  \AE{-0.1}
		  \AA{3}
	\end{scope}
\end{tikzpicture}
\end{center}

This gives for $\beta\gamma\epsilon$ and $\beta\delta\epsilon$:

\begin{center}
\begin{tikzpicture}[scale=0.5]
	\Aalpha{0}
	\Abeta{2.9}
	\Aalpha{5}
	\Adelta{7.9}
	\Aalpha{13}
	\Aepsilon{15.9}
	\Aalpha{18}
	\begin{scope}[yshift=-1cm]
		  \AA{0}
		  \AC{4.9}
		  \AA{13}
		  \AB{17.9}
	\end{scope}
	\begin{scope}[yshift=1cm]
		  \AE{-0.1}
		  \AA{3}
		  \AC{7.9}
		  \AA{16}
	\end{scope}
\end{tikzpicture}
\end{center}

\begin{center}
\begin{tikzpicture}[scale=0.5]
	\Aalpha{0}
	\Abeta{2.9}
	\Aalpha{5}
	\Agamma{7.9}
	\Aalpha{13}
	\Aepsilon{15.9}
	\Aalpha{18}
	\begin{scope}[yshift=-1cm]
		  \AA{0}
		  \AD{4.9}
		  \AA{13}
		  \AB{17.9}
	\end{scope}
	\begin{scope}[yshift=1cm]
		  \AE{-0.1}
		  \AA{3}
		  \AD{7.9}
		  \AA{16}
	\end{scope}
\end{tikzpicture}
\end{center}

And the piece de resistance $\beta\delta\gamma\epsilon$: 

\begin{center}
\begin{tikzpicture}[scale=0.4]
	\Aalpha{0}
	\Abeta{2.9}
	\Aalpha{5}
	\Adelta{7.9}
	\Aalpha{13}
	\Agamma{15.9}
	\Aalpha{21}
	\Aepsilon{23.9}
	\Aalpha{26}
	\begin{scope}[yshift=-1cm]
		  \AA{0}
		  \AC{4.9}
		  \AA{13}
		  \AE{17.9}
		  \AA{21}
		  \AB{25.9}
	\end{scope}
	\begin{scope}[yshift=1cm]
		  \AE{-0.1}
		  \AA{3}
		  \AB{7.9}
		  \AA{11}
		  \AD{15.9}
		  \AA{24}
	\end{scope}
\end{tikzpicture}
\end{center}
We now look at the transducer $T'$ we obtained with the preceding six pieces.
Note that $T' = T_n \circ T_{n+1}\circ T_n \circ \sigma^3$ where
$\sigma$ is the shift:
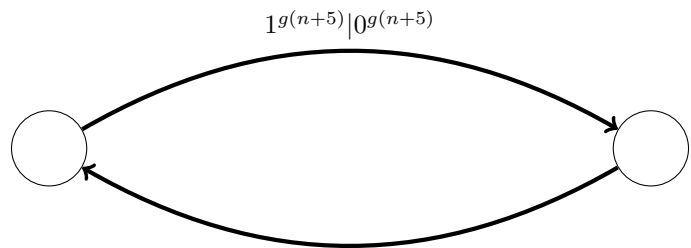
\begin{center}
	
\begin{tikzpicture}[->,node distance=5cm,auto,initial text=]
        
        \tikzstyle{every state}=[minimum size=10mm]

\node[state] (A) at (0,0)                   {};
\node[state] (F) at (8,0) {};
\path[ultra thick] (A) edge[bend left]   node[above] {$1^{g(n+5)}|0^{g(n+5)}$} (F);

\path[ultra thick] (F) edge[bend left]   node[below] {
$  \begin{array}{l@{|}l}
 1^{g(n+4)}&(110)0^{g(n+4)-3}\\
 1^{g(n+4)-3}(100)&0^{g(n+4)}\\
 1^{g(n+6)}&0^{g(n+5)}(111)0^{g(n+4)-3}\\
 1^{g(n+4)-3}(111)1^{g(n+5)}&0^{g(n+6)}\\
 1^{g(n+6)-3}(100)1^{g(n+4)}&0^{g(n+4)}(110)0^{g(n+6)-3}\\
\end{array}$} (A);
\end{tikzpicture}
\end{center}	

We recognize $T_{n+3}$ up to a shift of $3$, which proves the Theorem.

\section{End of the proof}
\label{sec:end}
\subsection{Aperiodicity of \texorpdfstring{$\ws$}{T}}
\begin{proposition}
  There are no words $u, v$ s.t.
$ u (T_{n+1} \circ T_{n} \circ T_{n+1}
\circ T_{n} \circ T_{n+1}) v$
\end{proposition}
\begin{proof}
By the previous section, $T_n$, when bordered by $T_{n+1}$ on both
sides, can be rewritten as concatenations of blocks of the following
five types:
$\beta\gamma\delta$, $\epsilon\gamma\delta$, $\beta\delta\gamma\epsilon\gamma\delta$,
$\beta\gamma\epsilon\gamma\delta$ and $\beta\delta\epsilon\gamma\delta$.

However, as $T_{n+1} \circ T_{n} \circ T_{n+1} \circ T_{n} \circ
T_{n+1} = T_{n+3} \circ T_n \circ T_{n+1}$, the block
$\epsilon\gamma\delta$  (and any block containing
it) cannot appear in the execution of the transducer $T_n$, as $T_{n+3}$ does not produce any input where
$100$ and $000$ are close enough.
So the only possible block remaining is $\beta\gamma\delta$.
But $T_{n+3}$ does not produce any input where $000$ and $000$ are
at distance $g(n+6)$.
\end{proof}

\begin{proposition}
Let $n \geq -2$. Any tiling of the plane by $\ws_D$ can be divided into strips of
vertical width $g(n), g(n+1)$ or $g(n+2)$ so that each strip is a tiling by
$\ws_{u_n},
          \ws_{u_{n+1}}$ or $\ws_{u_{n+2}}$.
        \end{proposition}  

\begin{proof}[Proof of the Proposition]
  The proof is by induction on $n$. The result is trivial for $n = -2,
  -1$, and true for $n = 0$ by Section~\ref{ssec:seq}.

Now suppose the result holds true for $n$. Consider a tiling of the plane by
$\ws_D$.
This tiling can be divided into strips that correspond to tilings by 
$\ws_{u_n},\ws_{u_{n+1}}$ or $\ws_{u_{n+2}}$.

By Proposition~\ref{prop:010}, the words in each row are elements of
$W$.
We can therefore replace each strip $\ws_{u_i}$ by $T_i$ to obtain a
tiling of the plane by $T_n \cup T_{n+1} \cup T_{n+2}$.
It is easy to see, given the inputs of these transducers that, in such
a tiling, each row corresponding to the transducer $T_n$ is surrounded
by rows corresponding to the transducer $T_{n+1}$.
As a consequence, each strip corresponding to $\ws_{u_{n}}$ is
surrounded by strips corresponding to $\ws_{u_{n+1}}$.

By the previous proposition, there are no words $u,v \in W$ s.t.
$u (T_{n+1} \circ T_n \circ T_{n+1} \circ T_n \circ T_{n+1}) v$.
As a consequence, there are no words $u,v \in W$ s.t.
$u (\ws_{u_{n+1}} \circ \ws_{u_{n}} \circ \ws_{u_{n+1}} \circ
\ws_{u_{n}} \circ \ws_{u_{n+1}}) v$.
Therefore, in the dividing of the plane by strips, we do not have 5
consecutive  strips of the words $\ws_{u_{n+1}}, \ws_{u_n},
\ws_{u_{n+1}}, \ws_{u_{n}}, \ws_{u_{n+1}}$

We can therefore replace each occurrence of 3 consecutive strips
$\ws_{u_{n+1}},$ $\ws_{u_{n}},$ $\ws_{u_{n+1}}$ by $\ws_{u_{n+3}}$ as no
occurrences overlap. After doing this, no occurrence of $\ws_{u_{n}}$
remains, which ends the proof.
\end{proof}

\begin{corollary}
The Wang set	$\ws_D = \ws_\texttt{a} \cup \ws_\texttt{b}$ is
aperiodic.

Furthermore, the set of words $u \in \{\texttt{a},\texttt{b}\}^\star$, s.t. the sequence of
transducers $\ws_u$ appears in a tiling of the plane, is exactly the set
of factors of the Fibonacci word, i.e., the set of factors of sturmian words of slope $1/\phi$, for $\phi$ the golden mean.

Biinfinite words $u \in \{a,b\}^\mathbb{Z}$, s.t $\ws_u$
which represents a valid tiling of the plane, are exactly the sturmian words of slope $1/\phi$. \end{corollary}	
See~\cite{berstelseebold} for some references on sturmian words.

\begin{proof}
First, note that, for all $n$, the transducer $T_n$ contains a
biinfinite path. In particular, there exists $u,v \in W$ s.t $uT_n v$ and
therefore s.t. $u \ws_{u_{n}} v$. We have then, for all $n$, a 
tiling of $g(n)$ consecutive rows by $\ws_D$. By compactness, there
exists a tiling of the plane by $\ws_D$.

Now consider any tiling by $\tau_D$. Let $v$ be the word over the
alphabet $\{\texttt{a},\texttt{b}\}$ s.t. $v_i = \texttt{a}$ if the $i$-th row of the tiling
corresponds to $\ws_\texttt{a}$  and $v_i = \texttt{b}$ otherwise.

By the previous proposition, any tiling by $\tau_D$ can be decomposed
into tilings by $\tau_{u_{n}}, \tau_{u_{n+1}}, \tau_{u_{n+2}}$ for all
$n$, which implies that the word $v$ can be written as a concatenation
of $u_{n}$, $u_{n+1}$ and $u_{n+2}$.

The sequence of words $u_n$ we defined is the sequence of singular
factors of the Fibonacci word (see for example~\cite{Wen}). Thus, $v$ has the same set of factors as the Fibonacci word. In
particular, $v$ is not periodic.
\end{proof}

\def\slope{1/(\phi+2)} \begin{corollary}
The Wang set $\ws$ is aperiodic.
Furthermore, the set of words $u \in \{0,1\}^\star$ s.t. the sequence of
transducers $\ws_u$ appears in a tiling of the plane is exactly the set
of factors of sturmian words of slope $\slope$, for $\phi$ the golden mean.

The set of biinfinite words $u \in \{0,1\}^\mathbb{Z}$ s.t $\ws_u$
which represents a valid tiling of the plane  are exactly the sturmian words  of slope $\slope$. \end{corollary}	
\begin{proof}
Let $\psi$ be the morphism $\texttt{a} \mapsto 10000, \texttt{b} \mapsto 1000$.
 The set of all words $u \in \{0,1\}^\mathbb{Z}$ that can appear in a
 tiling of the whole plane are exactly the image by $\psi$ of the
 sturmian words over the alphabet $\{a,b\}$ of slope $1/\phi$.
 
 It is well known that the image of a sturmian word by $\psi$ is again
 a sturmian word, see~\cite[Corollary 2.2.19]{berstelseebold}, where  $\psi = \tilde G^3 D$ (with $\{a,b\}$ instead of $\{0,1\}$ as  input alphabet).
 The derivation of the slope is routine. 
\end{proof}

\subsection{Aperiodicity of \texorpdfstring{$\ws'$}{T'}}
Recall that $\ws'$ is the Wang set from Figure~\ref{fig:ws4}.
This Wang set is obtained from $\ws$, by merging two vertical colors:
0 and 4 in $\ws$ become 0 in $\ws'$. Thus, every tiling of $\ws$ can be
turned into a tiling of $\ws'$, and therefore $\ws'$ tiles the plane. We will show below that every tiling of $\ws'$ can be turned into a tiling of $\ws$, and thus every tiling of $\ws'$ is aperiodic.

$\ws'$ is the union of two Wang sets $\ws'_0$ and $\ws'_1$ of respectively 9 and 2 tiles.
The following facts can be easily checked by computer.
For $w\in \{0,1\}^*\setminus\{\epsilon\}$, let $\ws'_{w} = {\ws'_{w[1]} \circ \ws'_{w[2]} \circ \ldots \ws'_{w[\vert w\vert]}}$.
\begin{fact}
\sloppy The transducers $\sc{\ws'_{111}}$, $\sc{\ws'_{101}}$, $\sc{\ws'_{1001}}$, $\sc{\ws'_{1000001}}$, $\sc{\ws'_{10000001}}$, $\sc{\ws'_{100000001}}$, $\sc{\ws'_{000000000}}$, $\sc{\ws'_{000011}}$, $\sc{\ws'_{110000}}$ and $\sc{\ws'_{1100011}}$ are empty.
\end{fact}

Thus, if $t$ is a tiling by $\ws'$, then there exists a biinfinite binary word $w\in \{1000,10000,100011000,$ $100000000\}^{\mathbb{Z}}$ such that $t(x,y)\in T(\ws'_{w[y]})$ for every $x,y\in\mathbb{Z}$.

Let $\ws'_A = \sc{\ws'_{1000} \cup \ws'_{10000} \cup {\ws'_{100000000} \cup {\ws'_{100011000}}}}$.
As before, $\ws'_A$ has unused transitions (those which write $2$ or $3$). Once deleted, and then once having deleted states which cannot appear in a tiling of a row, we obtain $\ws'_B$.
$\ws'_B$ has 4 connected components: two were already present in $\ws$: $\ws_\texttt{a}$ and $\ws_\texttt{b}$, the third one $\ws_c$ is a subset of $\ws'_{100000000}$, and the last one $\ws_d$ is a subset of $\ws'_{100011000}$.

\begin{proposition}\label{fa:cdab}
$\ws'_{11}$ is isomorphic to a subset of $\ws'_{01}$, and $\ws'_{100000}$ is isomorphic to a subset of $\ws'_{100001}$.
\end{proposition}
\begin{proof}
$\ws'_{11}$ is the transducer with one state, which reads $1$ and writes $2$. $\ws'_{01}$ also has a loop that reads $1$ and writes $2$: the transition $(02,02,1,2)$.
$\ws'_{100000}$ and $\ws'_{100001}$ are depicted in Figure~\ref{fig:w4100001} (in a compact form). $\ws'_{100000}$ is isomorphic to the subset of $\ws'_{100001}$ drawn in bold.
\end{proof}

\begin{figure}
\centering
\raisebox{-0.5\height}{
\begin{tikzpicture}[auto,scale=1.6]
\tikzstyle{every state}=[shape=ellipse,inner sep=0mm,font=,scale=0.7]
\tikzstyle{trans}=[->,thin]
\tikzstyle{tr}=[inner sep=.5mm,font=,scale=0.7]
\node[state] (n210302) at (2.912000,4.704000) {210302};
\node[state] (n210332) at (7.168000,7.168000) {210332};
\node[state] (n211032) at (3.808000,6.496000) {211032};
\node[state] (n211302) at (6.048000,5.600000) {211302};
\node[state] (n213002) at (7.168000,4.704000) {213002};
\node[state] (n213102) at (2.912000,7.168000) {213102};
\node[state] (n213302) at (6.048000,6.496000) {213302};
\path[trans] (n210302) edge[thick] node[tr] {$1|2$} (n213002);
\path[trans] (n210332) edge node[tr] {$1|0$} (n213002);
\path[trans] (n210302) edge node[tr] {$1|2$} (n213102);
\path[trans] (n210332) edge node[tr] {$1|0$} (n213102);
\path[trans] (n211032) edge node[tr] {$1|0$} (n210302);
\path[trans] (n211302) edge[bend left=25,thick] node[tr] {$1|2$} (n213302);
\path[trans] (n213002) edge[thick] node[tr] {$1|2$} (n211302);
\path[trans] (n213102) edge node[tr] {$1|0$} (n211032);
\path[trans] (n213302) edge node[tr] {$11|00$} (n210332);
\path[trans] (n211302) edge[thick] node[tr] {$000|222$} (n210302);
\path[trans] (n213302) edge[bend left=25,thick] node[tr] {$00|22$} (n211302);
\path[trans] (n213302) edge node[tr] {$00|20$} (n211032);
\end{tikzpicture}
}
\hspace{1cm}
\raisebox{-0.5\height}{
\begin{tikzpicture}[auto,scale=1.6]
\tikzstyle{every state}=[shape=ellipse,inner sep=0mm,font=,scale=0.7]
\tikzstyle{trans}=[->,thin]
\tikzstyle{tr}=[inner sep=.5mm,font=,scale=0.7]
\node[state] (n211301) at (0.000000,0.000000) {211301};
\path[trans] (n211301) edge[loop above] node[tr] {$100|222$} (n211301);
\path[trans] (n211301) edge[loop below] node[tr] {$00011|22222$} (n211301);
\end{tikzpicture}
}
\caption{$\ws'_{100001}$ (left) and $\ws'_{100000}$ (right).}
\label{fig:w4100001}
\end{figure}

\begin{corollary}\label{cor:cdab}
$\ws_c$ and $\ws_d$ are both isomorphic to a subset of $\ws_\texttt{a} \circ \ws_\texttt{b}$.
\end{corollary}

A tiling of $\ws'_B$ can thus be turned into a tiling of $\ws_B$, by substituting every tile from $\ws_c$ (resp. $\ws_d$) by two tiles, one from $\ws_\texttt{a}$ and one from $\ws_\texttt{b}$.

\begin{theorem}
The Wang set $\ws'$ is aperiodic.
\end{theorem}
\begin{proof}
The Wang set $\ws'$ is aperiodic if and only if $\ws'_B$ is aperiodic. 
Suppose that $\ws'_B$ is not aperiodic. We know that $\ws'$, and thus
$\ws'_B$ tile the plane. Take a periodic tiling by $\ws'_B$. This
tiling can be turned into a tiling of $\ws_B$ by the Corollary~\ref{cor:cdab}. Thus $\ws_B$ has a periodic tiling, contradiction.
\end{proof}

\subsection{A third aperiodic set \texorpdfstring{$\ws''$}{T''}}\label{sec:third}
During our research, we also find a third aperiodic set $\ws''$ of 11 Wang tiles (Figure~\ref{fig:ws4third}).
As for the two others, $\ws''$ is the union of two Wang sets, $\ws''_0$ and $\ws''_1$, of respectively 9 and 2 tiles.
For $w\in \{0,1\}^*\setminus\{\epsilon\}$, let $\ws''_{w} = {\ws''_{w[1]} \circ \ws''_{w[2]} \circ \ldots \ws''_{w[\vert w\vert]}}$.

\begin{fact}
The transducers $\sc{\ws''_{11}}$, $\sc{\ws''_{101}}$, $\sc{\ws''_{1001}}$ and $\sc{\ws''_{00000}}$ are empty.
Therefore, if $t$ is a tiling by $\ws''$, there exists a biinfinite binary word $w\in \{1000,10000\}^{\mathbb{Z}}$ such that $t(x,y)\in T(\ws''_{w[y]})$ for every $x,y\in\mathbb{Z}$.
\end{fact}

\begin{figure}
\center

\begin{tikzpicture}[scale=0.9]
\tikzstyle{every state}=[]
\tikzstyle{tr}=[inner sep=.5mm]
\tikzstyle{trans}=[->,thick,auto]
\node[state] (n0) at (1.600000,0.000000) {0};
\node[state] (n1) at (1.600000,3.200000) {1};
\node[state] (n2) at (3.200000,1.600000) {2};
\node[state] (n3) at (0.000000,1.600000) {3};
\path[trans] (n0) edge[loop right] node[tr] {$1|0$} (n0);
\path[trans] (n0) edge[bend left=10] node[tr] {$2|1$} (n3);
\path[trans] (n1) edge[bend left=10] node[tr] {$2|2$} (n0);
\path[trans] (n1) edge[bend left=10] node[tr] {$2|3$} (n3);
\path[trans] (n3) edge[bend left=10] node[tr] {$1|1$} (n0);
\path[trans] (n3) edge[bend left=10] node[tr,align=left] {$1|1$\\$2|2$} (n1);
\path[trans] (n3) edge[loop left] node[tr,align=left] {$3|1$\\$4|2$} (n3);
\path[trans] (n2) edge[loop right] node[tr,align=left] {$1|4$\\$0|2$} (n2);
 \end{tikzpicture}

\medskip

\center
\begin{tikzpicture}
\draw (0.000000,0.000000) -- (1.000000,0.000000) ;
\draw (0.000000,0.000000) -- (1.000000,1.000000) ;
\draw (0.000000,0.000000) -- (0.000000,1.000000) ;
\draw (1.000000,0.000000) -- (1.000000,1.000000) ;
\draw (0.000000,1.000000) -- (1.000000,1.000000) ;
\draw (0.000000,1.000000) -- (1.000000,0.000000) ;
\draw (0.420000,0.500000) node[left]{$0$} ;
\draw (0.570000,0.500000) node[right]{$0$} ;
\draw (0.500000,0.540000) node[above]{$0$} ;
\draw (0.500000,0.460000) node[below]{$1$} ;
\draw (1.300000,0.000000) -- (2.300000,0.000000) ;
\draw (1.300000,0.000000) -- (2.300000,1.000000) ;
\draw (1.300000,0.000000) -- (1.300000,1.000000) ;
\draw (2.300000,0.000000) -- (2.300000,1.000000) ;
\draw (1.300000,1.000000) -- (2.300000,1.000000) ;
\draw (1.300000,1.000000) -- (2.300000,0.000000) ;
\draw (1.720000,0.500000) node[left]{$0$} ;
\draw (1.870000,0.500000) node[right]{$3$} ;
\draw (1.800000,0.540000) node[above]{$1$} ;
\draw (1.800000,0.460000) node[below]{$2$} ;
\draw (2.600000,0.000000) -- (3.600000,0.000000) ;
\draw (2.600000,0.000000) -- (3.600000,1.000000) ;
\draw (2.600000,0.000000) -- (2.600000,1.000000) ;
\draw (3.600000,0.000000) -- (3.600000,1.000000) ;
\draw (2.600000,1.000000) -- (3.600000,1.000000) ;
\draw (2.600000,1.000000) -- (3.600000,0.000000) ;
\draw (3.020000,0.500000) node[left]{$1$} ;
\draw (3.170000,0.500000) node[right]{$0$} ;
\draw (3.100000,0.540000) node[above]{$2$} ;
\draw (3.100000,0.460000) node[below]{$2$} ;
\draw (3.900000,0.000000) -- (4.900000,0.000000) ;
\draw (3.900000,0.000000) -- (4.900000,1.000000) ;
\draw (3.900000,0.000000) -- (3.900000,1.000000) ;
\draw (4.900000,0.000000) -- (4.900000,1.000000) ;
\draw (3.900000,1.000000) -- (4.900000,1.000000) ;
\draw (3.900000,1.000000) -- (4.900000,0.000000) ;
\draw (4.320000,0.500000) node[left]{$1$} ;
\draw (4.470000,0.500000) node[right]{$3$} ;
\draw (4.400000,0.540000) node[above]{$3$} ;
\draw (4.400000,0.460000) node[below]{$2$} ;
\draw (5.200000,0.000000) -- (6.200000,0.000000) ;
\draw (5.200000,0.000000) -- (6.200000,1.000000) ;
\draw (5.200000,0.000000) -- (5.200000,1.000000) ;
\draw (6.200000,0.000000) -- (6.200000,1.000000) ;
\draw (5.200000,1.000000) -- (6.200000,1.000000) ;
\draw (5.200000,1.000000) -- (6.200000,0.000000) ;
\draw (5.620000,0.500000) node[left]{$3$} ;
\draw (5.770000,0.500000) node[right]{$0$} ;
\draw (5.700000,0.540000) node[above]{$1$} ;
\draw (5.700000,0.460000) node[below]{$1$} ;
\draw (6.500000,0.000000) -- (7.500000,0.000000) ;
\draw (6.500000,0.000000) -- (7.500000,1.000000) ;
\draw (6.500000,0.000000) -- (6.500000,1.000000) ;
\draw (7.500000,0.000000) -- (7.500000,1.000000) ;
\draw (6.500000,1.000000) -- (7.500000,1.000000) ;
\draw (6.500000,1.000000) -- (7.500000,0.000000) ;
\draw (6.920000,0.500000) node[left]{$3$} ;
\draw (7.070000,0.500000) node[right]{$1$} ;
\draw (7.000000,0.540000) node[above]{$1$} ;
\draw (7.000000,0.460000) node[below]{$1$} ;
\draw (0.000000,1.300000) -- (1.000000,1.300000) ;
\draw (0.000000,1.300000) -- (1.000000,2.300000) ;
\draw (0.000000,1.300000) -- (0.000000,2.300000) ;
\draw (1.000000,1.300000) -- (1.000000,2.300000) ;
\draw (0.000000,2.300000) -- (1.000000,2.300000) ;
\draw (0.000000,2.300000) -- (1.000000,1.300000) ;
\draw (0.420000,1.800000) node[left]{$3$} ;
\draw (0.570000,1.800000) node[right]{$1$} ;
\draw (0.500000,1.840000) node[above]{$2$} ;
\draw (0.500000,1.760000) node[below]{$2$} ;
\draw (1.300000,1.300000) -- (2.300000,1.300000) ;
\draw (1.300000,1.300000) -- (2.300000,2.300000) ;
\draw (1.300000,1.300000) -- (1.300000,2.300000) ;
\draw (2.300000,1.300000) -- (2.300000,2.300000) ;
\draw (1.300000,2.300000) -- (2.300000,2.300000) ;
\draw (1.300000,2.300000) -- (2.300000,1.300000) ;
\draw (1.720000,1.800000) node[left]{$3$} ;
\draw (1.870000,1.800000) node[right]{$3$} ;
\draw (1.800000,1.840000) node[above]{$1$} ;
\draw (1.800000,1.760000) node[below]{$3$} ;
\draw (2.600000,1.300000) -- (3.600000,1.300000) ;
\draw (2.600000,1.300000) -- (3.600000,2.300000) ;
\draw (2.600000,1.300000) -- (2.600000,2.300000) ;
\draw (3.600000,1.300000) -- (3.600000,2.300000) ;
\draw (2.600000,2.300000) -- (3.600000,2.300000) ;
\draw (2.600000,2.300000) -- (3.600000,1.300000) ;
\draw (3.020000,1.800000) node[left]{$3$} ;
\draw (3.170000,1.800000) node[right]{$3$} ;
\draw (3.100000,1.840000) node[above]{$2$} ;
\draw (3.100000,1.760000) node[below]{$4$} ;
\draw (3.900000,1.300000) -- (4.900000,1.300000) ;
\draw (3.900000,1.300000) -- (4.900000,2.300000) ;
\draw (3.900000,1.300000) -- (3.900000,2.300000) ;
\draw (4.900000,1.300000) -- (4.900000,2.300000) ;
\draw (3.900000,2.300000) -- (4.900000,2.300000) ;
\draw (3.900000,2.300000) -- (4.900000,1.300000) ;
\draw (4.320000,1.800000) node[left]{$2$} ;
\draw (4.470000,1.800000) node[right]{$2$} ;
\draw (4.400000,1.840000) node[above]{$4$} ;
\draw (4.400000,1.760000) node[below]{$1$} ;
\draw (5.200000,1.300000) -- (6.200000,1.300000) ;
\draw (5.200000,1.300000) -- (6.200000,2.300000) ;
\draw (5.200000,1.300000) -- (5.200000,2.300000) ;
\draw (6.200000,1.300000) -- (6.200000,2.300000) ;
\draw (5.200000,2.300000) -- (6.200000,2.300000) ;
\draw (5.200000,2.300000) -- (6.200000,1.300000) ;
\draw (5.620000,1.800000) node[left]{$2$} ;
\draw (5.770000,1.800000) node[right]{$2$} ;
\draw (5.700000,1.840000) node[above]{$2$} ;
\draw (5.700000,1.760000) node[below]{$0$} ;
\end{tikzpicture}
 
\caption{The aperiodic Wang set $\ws''$.} \label{fig:ws4third}
\end{figure}

$\ws''_{1000}$ (resp. $\ws''_{10000}$) does not act exactly as $\ws_{1000}$ (resp. $\ws_{10000}$).
However, if we compose them with the shift transducer $\mathcal{S}$ (Figure~\ref{shift}), we get transducers equivalent to $\ws_C$.
Let $\ws''_A=\sc{(\ws''_{1000} \cup \ws''_{10000})\circ S}$.
It is easy to see that the composition with $\mathcal{S}$ does not change the aperiodic status.
$\ws''_A$ never reads $2$, $3$ nor $4$. 
Thus the transitions that write $2$, $3$ or $4$ are never used in a tiling by $\ws''_A$.
Let $\ws''_B$  (Figure~\ref{fig:2_wsb}) be the transducer $\ws''_A$ after removing these unused
transitions, and deleting states that cannot appear in a tiling of a row (i.e., sources and sinks).
Some states are bisimilar in $\ws''_B$. If we contract these states, we got $\ws''_C$ (Figure~\ref{fig:2_wsc}), which is isomorphic to $\ws_C$. Thus $\ws''$ is aperiodic.

\begin{figure}[h!]
\center

\begin{tikzpicture}[scale=0.9]
\tikzstyle{every state}=[]
\tikzstyle{tr}=[inner sep=.5mm]
\tikzstyle{trans}=[->,thick,auto]
\node[state] (n0) at (0,0) {0};
\node[state] (n1) at (2,0) {1};
\path[trans] (n0) edge[loop left] node[tr]  {$0|0$} (n0);
\path[trans] (n1) edge[bend left] node[tr]  {$1|0$} (n0);
\path[trans] (n0) edge[bend left] node[tr]  {$0|1$} (n1);
\path[trans] (n1) edge[loop right] node[tr]  {$1|1$} (n1);
 \end{tikzpicture}
\caption{The shift transducer $\mathcal{S}$.} 
\label{shift}
\end{figure}

\begin{figure}[h!]
\center

\begin{subfigure}[b]{.999\linewidth}
\begin{minipage}{.44\linewidth}
\center
\begin{tikzpicture}[auto,scale=.75]
\tikzstyle{every state}=[shape=ellipse,inner sep=0mm,font=\scriptsize,scale=0.6]
\tikzstyle{trans}=[->,thin]
\tikzstyle{tr}=[inner sep=0.2mm,font=\footnotesize,scale=0.6]
\node[state] (n203300) at (7.168000,4.928000) {203300};
\node[state] (n211300) at (5.824000,4.928000) {211300};
\node[state] (n213300) at (3.584000,2.240000) {213300};
\node[state] (n230300) at (8.412000,4.032000) {230300};
\node[state] (n230331) at (8.412000,3.332000) {230331};
\node[state] (n231000) at (8.808000,5.128000) {231000};
\node[state] (n231031) at (9.708000,4.528000) {231031};
\node[state] (n231131) at (7.182000,3.284000) {231131};
\node[state] (n231300) at (4.480000,3.284000) {231300};
\node[state] (n233000) at (3.584000,4.928000) {233000};
\node[state] (n233101) at (9.408000,2.240000) {233101};
\node[state] (n233301) at (5.824000,3.284000) {233301};
\path[trans] (n203300) edge node[tr]   {$0|0$} (n231000);
\path[trans] (n211300) edge node[tr]   {$0|0$} (n203300);
\path[trans] (n213300) edge node[tr]   {$0|0$} (n233000);
\path[trans] (n213300) edge node[tr]  {$0|1$} (n233101);
\path[trans] (n230300) edge[bend left=5] node[tr,swap]   {$1|0$} (n233000);
\path[trans] (n230300) edge[bend left=10] node[tr]  {$1|1$} (n233101);
\path[trans] (n230331) edge[bend right=5] node[tr]  {$1|0$} (n233000);
\path[trans] (n230331) edge node[tr,swap]  {$1|1$} (n233101);
\path[trans] (n231000) edge node[tr,swap]   {$1|0$} (n230300);
\path[trans] (n231031) edge node[tr]  {$1|0$} (n230300);
\path[trans] (n231131) edge node[tr,swap]  {$1|1$} (n230331);
\path[trans] (n231300) edge node[tr]   {$0|0$} (n213300);
\path[trans] (n231300) edge node[tr]   {$1|1$} (n233301);
\path[trans] (n233000) edge node[tr]   {$0|0$} (n211300);
\path[trans] (n233000) edge node[tr]   {$1|0$} (n231300);
\path[trans] (n233101) edge node[tr,swap]  {$1|1$} (n231031);
\path[trans] (n233301) edge node[tr]   {$1|1$} (n231131);
 \end{tikzpicture}
\end{minipage}
\begin{minipage}{.55\linewidth}
\center 
\begin{tikzpicture}[auto,scale=0.55]
\tikzstyle{every state}=[shape=ellipse,inner sep=0mm,font=\scriptsize,scale=0.6]
\tikzstyle{trans}=[->,thin]
\tikzstyle{tr}=[inner sep=0.2mm,font=\footnotesize,scale=0.6]
\node[state] (n20300) at (7.492000,4.732000) {20300};
\node[state] (n20331) at (7.492000,3.432000) {20331};
\node[state] (n21031) at (9.632000,4.032000) {21031};
\node[state] (n21131) at (5.824000,4.032000) {21131};
\node[state] (n21300) at (3.040000,6.920000) {21300};
\node[state] (n21331) at (2.140000,6.220000) {21331};
\node[state] (n23000) at (5.824000,5.920000) {23000};
\node[state] (n23031) at (5.824000,6.920000) {23031};
\node[state] (n23101) at (9.632000,6.920000) {23101};
\node[state] (n23131) at (2.240000,4.032000) {23131};
\node[state] (n23301) at (4.032000,4.032000) {23301};
\path[trans] (n20300) edge node[tr,swap]   {$0|0$} (n23000);
\path[trans] (n20300) edge node[tr]  {$0|1$} (n23101);
\path[trans] (n20331) edge node[tr]  {$0|0$} (n23000);
\path[trans] (n20331) edge node[tr,swap]  {$0|1$} (n23101);
\path[trans] (n21031) edge[pos=.3] node[tr]   {$0|0$} (n20300);
\path[trans] (n21131) edge node[tr,swap]  {$0|1$} (n20331);
\path[trans] (n21300) edge node[tr]   {$0|1$} (n23301);
\path[trans] (n21331) edge node[tr,swap]  {$0|1$} (n23301);
\path[trans] (n23000) edge node[tr]   {$0|0$} (n21300);
\path[trans] (n23000) edge node[tr]   {$1|1$} (n23301);
\path[trans] (n23031) edge node[tr,swap]  {$0|0$} (n21300);
\path[trans] (n23031) edge[bend right=5] node[tr,swap]  {$1|1$} (n23301);
\path[trans] (n23101) edge node[tr]   {$0|1$} (n21031);
\path[trans] (n23101) edge node[tr,swap]  {$1|1$} (n23031);
\path[trans] (n23131) edge node[tr]  {$0|1$} (n21331);
\path[trans] (n23301) edge node[tr]   {$0|1$} (n21131);
\path[trans] (n23301) edge node[tr,swap]   {$1|1$} (n23131);
 \end{tikzpicture}
\end{minipage}
\caption{Wang set $\ws''_B$.} 
\label{fig:2_wsb}
\vspace{.5cm}
\end{subfigure}

\begin{subfigure}[b]{.999\linewidth}
\begin{minipage}{.44\linewidth}
\center
\begin{tikzpicture}[auto,scale=.75]
\tikzstyle{every state}=[shape=ellipse,inner sep=0mm,font=\scriptsize,scale=0.6]
\tikzstyle{trans}=[->,thin]
\tikzstyle{tr}=[inner sep=0.2mm,font=\footnotesize,scale=0.6]
\node[state] (n203300) at (7.168000,4.928000) {203300};
\node[state] (n211300) at (5.824000,4.928000) {211300};
\node[state] (n213300) at (3.584000,2.240000) {213300};
\node[state] (n230300) at (8.512000,3.584000) {230300};
\node[state] (n231000) at (9.408000,4.928000) {231000};
\node[state] (n231131) at (7.182000,3.584000) {231131};
\node[state] (n231300) at (4.480000,3.584000) {231300};
\node[state] (n233000) at (3.584000,4.928000) {233000};
\node[state] (n233101) at (9.408000,2.240000) {233101};
\node[state] (n233301) at (5.824000,3.584000) {233301};
\path[trans] (n203300) edge node[tr]         {$0|0$} (n231000);
\path[trans] (n211300) edge node[tr]         {$0|0$} (n203300);
\path[trans] (n213300) edge node[tr]         {$0|0$} (n233000);
\path[trans] (n213300) edge node[tr]        {$0|1$} (n233101);
\path[trans] (n230300) edge[out=150,in=-20] node[tr,swap,pos=.15]         {$1|0$} (n233000);
\path[trans] (n230300) edge node[tr]        {$1|1$} (n233101);
\path[trans] (n231000) edge node[tr,swap]         {$1|0$} (n230300);
\path[trans] (n231131) edge node[tr,swap]         {$1|1$} (n230300);
\path[trans] (n231300) edge node[tr]         {$0|0$} (n213300);
\path[trans] (n231300) edge node[tr]         {$1|1$} (n233301);
\path[trans] (n233000) edge node[tr]         {$0|0$} (n211300);
\path[trans] (n233000) edge node[tr]         {$1|0$} (n231300);
\path[trans] (n233101) edge node[tr]         {$1|1$} (n231000);
\path[trans] (n233301) edge node[tr]         {$1|1$} (n231131);
 \end{tikzpicture}
\end{minipage}
\begin{minipage}{.55\linewidth}
\center 
\begin{tikzpicture}[auto,scale=0.55]
\tikzstyle{every state}=[shape=ellipse,inner sep=0mm,font=\scriptsize,scale=0.6]
\tikzstyle{trans}=[->,thin]
\tikzstyle{tr}=[inner sep=0.2mm,font=\footnotesize,scale=0.6]
\node[state] (n20300) at (7.700000,4.032000) {20300};
\node[state] (n21031) at (9.632000,4.032000) {21031};
\node[state] (n21131) at (5.824000,4.032000) {21131};
\node[state] (n21300) at (2.240000,6.720000) {21300};
\node[state] (n23000) at (5.824000,6.720000) {23000};
\node[state] (n23101) at (9.632000,6.720000) {23101};
\node[state] (n23131) at (2.240000,4.032000) {23131};
\node[state] (n23301) at (4.032000,4.032000) {23301};
\path[trans] (n20300) edge node[tr]   {$0|0$} (n23000);
\path[trans] (n20300) edge node[tr]  {$0|1$} (n23101);
\path[trans] (n21031) edge[pos=.3] node[tr,swap]   {$0|0$} (n20300);
\path[trans] (n21131) edge[swap] node[tr,swap]   {$0|1$} (n20300);
\path[trans] (n21300) edge node[tr]   {$0|1$} (n23301);
\path[trans] (n23000) edge[pos=.3] node[tr]   {$0|0$} (n21300);
\path[trans] (n23000) edge node[tr]   {$1|1$} (n23301);
\path[trans] (n23101) edge[swap] node[tr]   {$0|1$} (n21031);
\path[trans] (n23101) edge node[tr,pos=.7]   {$1|1$} (n23000);
\path[trans] (n23131) edge[pos=.75] node[tr]   {$0|1$} (n21300);
\path[trans] (n23301) edge node[tr]   {$0|1$} (n21131);
\path[trans] (n23301) edge node[tr,swap]   {$1|1$} (n23131);
 \end{tikzpicture}
\end{minipage}
\caption{Wang set $\ws''_C$, the simplification of $\ws''_B$ by bisimulation.} 
\label{fig:2_wsc}
\vspace{.5cm}
\end{subfigure}

\end{figure}

\subsection{Remarks}

The reader may regret that our substitutive system starts
from $\ws_\texttt{b} \cup \ws_\texttt{aa} \cup \ws_{\texttt{bab}}$ and not from
$\ws_\texttt{a} \cup \ws_\texttt{b} \cup \ws_\texttt{aa}$, or even from $\ws_\texttt{a} \cup
\ws_\texttt{b}$. We do not know if this is possible. Our definition of
$T_n$ certainly does not work for $n = -1$, and the natural
generalization of it is not equivalent to $\ws_\texttt{a}$.
This is somewhat obvious, as $T_n$ (for $n \geq 0$) cannot be
composed with itself, whereas $\ws_\texttt{a}$ should be composed with
itself to obtain $\ws_\texttt{aa}$.

$\ws_\texttt{a}$ and $\ws_\texttt{b}$ both have the property of being
 time symmetric: if we reverse the directions of all edges,
exchange inputs and outputs, and exchange 0 and 1, we obtain an
equivalent transducer (it is obvious for $\ws_\texttt{b}$ and becomes
obvious for $\ws_\texttt{a}$ if we write it in a compact form without the states
$h$ and $g$). This property was used to simplify the proof that
the sequence $(T_n)$ is a recursive sequence, but we do not know
whether it can be used to simplify the entire proof.

While we gave a sequence of transducers $T_n$, it is, of
  course, possible to give another sequence of transducers, say $U_n$,
  which are equivalent to $T_n$, and have therefore the same properties.
  Our sequence $T_n$ has nice properties, in particular the
  symmetry explained above and its short number of transitions,
  but it has the drawback that the substitution, once seen
  geometrically, has small bumps due to the fact that the tiles are
  aligned only up to $\pm 3$.
  It is possible to find a sequence $U_n$ for which this does not
  appear, by splitting some transitions of $T_n$ into transitions
  of size $g(k)$ and transitions of size exactly $3$. However,
   this complicates the proof that the sequence is recursive. We think our
  sequence $T_n$ reaches an acceptable compromise.

  We do not know whether it is possible to obtain the result
  directly on the original tileset $\ws$ rather than $\ws_D$.
  A difficulty for this approach would be that $\ws$ is not purely substitutive (due 
  in part to the fact that no sturmian word of slope
  $\slope$ is purely morphic). At best, we could obtain that
   tilings by $\ws$ are images by some map $\phi$ of 
  some substitutive tilings (which is more or less what we obtain in
  our proof).

\section{Conclusion}\label{sec:conclusion}

We have shown that there is an aperiodic set of 11 Wang tiles, and that it is the smallest possible. Moreover, the set uses only 4 colors, and this is also the minimum possible among all aperiodic Wang sets.

During our research, we also obtained a large number of Wang sets with 11
tiles which are candidates for aperiodicity. These candidates are available on the repository.
The reader might ask
  why we choose to investigate this particular set, $\ws$. The reason is
  that it is very easy for a
  computer to produce the transducer for $\ws^k$, even for large
  values of $k$ ($k \sim 1000$).
  In contrast, for almost all other tilesets, we were not able to
  reach even $k = 30$. This suggested this tileset had some
  particular structure.
We will not give here more details on all our candidates, but
  we will say that a large number of them
  are tilesets corresponding to Kari's method, with one or more tiles
  omitted. With the method we described, we were
  able to prove that some of them do not tile the plane, but the
  method did not work on all of them.
  For now, we have found only three tilesets: the ones presented in this article, which were likely to be
  substitutive or nearly substitutive.

  Experimental results tend to support the following conjecture
  \begin{conjecture}
    Let $f(n)$ be the smallest $k$ s.t. every Wang set of size $n$ that
    does not tile the plane does not tile a square of size $k$.
    Let $g(n)$ be the smallest $k$ s.t. every Wang set of size $n$ that
    tiles the plane periodically does so with a period $p \leq k$.

    Then $g(n) \leq f(n)$ for all $n$.
  \end{conjecture}
\newbox\toto
\newcommand\ACA{\begin{tikzpicture}[scale=0.1]  
\draw (0,0) rectangle (2,2);
\end{tikzpicture}
}

\newcommand\ADA{\begin{tikzpicture}[scale=0.1]
  \draw (0,0) rectangle (6,2);
  \draw[fill=red,red] (0,1.5) rectangle (3,2);
\end{tikzpicture}
}

\newcommand\AEA{\begin{tikzpicture}[scale=0.1]
  \draw (0,0) rectangle (11,2);
  \draw[fill=green,green] (5,1.5) rectangle (8,2);
\end{tikzpicture}
}

\newcommand\AAA{\begin{tikzpicture}[scale=0.1]
  \draw (0,0) rectangle (11,2);
  \draw[fill=blue,blue] (3,0) rectangle (6,0.5);
\end{tikzpicture}
}

\newcommand\ABA{\begin{tikzpicture}[scale=0.1]
  \draw (0,0) rectangle (6,2);
  \draw[fill=black] (6,0) rectangle (3,.5);
\end{tikzpicture}
}

\newcommand\ACB{\begin{tikzpicture}[scale=0.1]
  \draw (0,0) rectangle (5,3);
\end{tikzpicture}
}

\newcommand\ADB{\begin{tikzpicture}[scale=0.1]
  \draw (0,0) rectangle (8,3);
  \draw[fill=black] (0,2.5) rectangle (3,3);
\end{tikzpicture}
}

\newcommand\AEB{\begin{tikzpicture}[scale=0.1]
  \draw (0,0) rectangle (16,3);
  \draw[fill=blue,blue] (8,2.5) rectangle (11,3);
\end{tikzpicture}
}

\newcommand\AAB{\begin{tikzpicture}[scale=0.1]
  \draw (0,0) rectangle (16,3);
  \draw[fill=green,green] (5,0) rectangle (8,0.5);
\end{tikzpicture}
}

\newcommand\ABB{\begin{tikzpicture}[scale=0.1]
  \draw (0,0) rectangle (8,3);
  \draw[fill=red,red] (5,0) rectangle (8,.5);
\end{tikzpicture}
}

\newcommand\ACC{\begin{tikzpicture}[scale=0.1]
  \draw (0,0) rectangle (10,5);
\end{tikzpicture}
}

\newcommand\ADC{\begin{tikzpicture}[scale=0.1]
  \draw (0,0) rectangle (11,5);
  \draw[fill=red,red] (0,4.5) rectangle (3,5);
\end{tikzpicture}
}

\newcommand\AEC{\begin{tikzpicture}[scale=0.1]
  \draw (0,0) rectangle (24,5);
  \draw[fill=green,green] (13,4.5) rectangle (16,5);
\end{tikzpicture}
}

\newcommand\ABC{\begin{tikzpicture}[scale=0.1]
  \draw (0,0) rectangle (11,5);
  \draw[fill=black] (11,0) rectangle (8,.5);
\end{tikzpicture}
}

\newcommand\AAC{\begin{tikzpicture}[scale=0.1]
  \draw (0,0) rectangle (24,5);
  \draw[fill=blue,blue] (8,0) rectangle (11,0.5);
\end{tikzpicture}
}
 
\def\definetile#1{\setbox\toto=\hbox{\csname #1\endcsname}
\immediate\pdfxform\toto
\expandafter\xdef\csname metatile#1\endcsname{\the\pdflastxform}
}
\definetile{AEA}
\definetile{ACA}
\definetile{ADA}
\definetile{AEA}
\definetile{AAA}
\definetile{ABA}
\definetile{ACB}
\definetile{ADB}
\definetile{AEB}
\definetile{AAB}
\definetile{ABB}
\definetile{ACC}
\definetile{ADC}
\definetile{AEC}
\definetile{ABC}
\definetile{AAC}
\setlength\unitlength{.1cm}

\begin{figure}
	\begin{center}

}
\end{center}
\caption{Representation of the meta-tile $\gamma$ (resp. $C$ if $n$ is
  odd) of  $T_{n}$ as tiles of $T_0 \uplus T_1 \uplus T_2$ for $ n=  0,1,2,3,4,5,6,7$.}
\end{figure}

\begin{figure}
\begin{center}
	%
\end{center}
	\caption{A fragment of a tiling by the transducers $T_0, T_1, T_2$.}
\end{figure}

\expandafter\xdef\csname color 0\endcsname{white}
\expandafter\xdef\csname color 1\endcsname{red}
\expandafter\xdef\csname color 2\endcsname{blue}
\expandafter\xdef\csname color 3\endcsname{green}
\expandafter\xdef\csname color 4\endcsname{red}

\def\definetile#1#2#3#4#5{\setbox\toto=\hbox{\begin{tikzpicture}[scale=.4]
\fill[color=\csname color #4\endcsname] (0,0) -- (0.5, 0.5) -- (1, 0) -- cycle;
\fill[color=\csname color #5\endcsname] (0,1) -- (0.5, 0.5) -- (1, 1) -- cycle;
\fill[color=\csname color #3\endcsname] (0,0) -- (0.5, 0.5) -- (0, 1) -- cycle;
\fill[color=\csname color #2\endcsname] (1,0) -- (0.5, 0.5) -- (1, 1) -- cycle;
\draw (0,0) rectangle (1,1);
\draw (0,0) -- (1,1);
\draw (0,1) -- (1,0);
\end{tikzpicture}}
\immediate\pdfxform\toto
\expandafter\xdef\csname tile#1\endcsname{\the\pdflastxform}
}

\definetile{a}{0}{0}{1}{0}
\definetile{b}{3}{0}{2}{1}
\definetile{c}{0}{1}{2}{2}
\definetile{d}{1}{1}{0}{2}
\definetile{e}{3}{1}{2}{3}
\definetile{f}{0}{3}{1}{1}
\definetile{g}{1}{3}{1}{1}
\definetile{h}{1}{3}{2}{2}
\definetile{i}{3}{3}{3}{1}
\definetile{x}{2}{2}{1}{0}
\definetile{y}{2}{2}{0}{2}

\setlength\unitlength{.4cm}

\begin{figure}
\begin{center}

\end{center}

	\caption{A fragment of a tiling by $\ws'$, with (0,1,2,3)=(white,red,blue,green).}
\end{figure}

\section*{Aknowledgements}
We gratefully acknowledge support from the PSMN (Pôle Scientifique de Modélisation Numérique) of the ENS de Lyon for the computing resources.
We also thank Jorge Pereda and anonymous reviewers for their careful reading of the manuscript and for their comments.





\bibliographystyle{plain}
\bibliography{../../../Biblio/biblio,bib}

\begin{aicauthors}
\begin{authorinfo}[ej]
  Emmanuel Jeandel\\
  Université de Lorraine, CNRS, Inria, LORIA\\
  Nancy, France\\
  emmanuel.jeandel\imageat{}loria\imagedot{}fr\\
  \url{https://members.loria.fr/EJeandel/}
\end{authorinfo}
\begin{authorinfo}[mr]
  Micha\"el Rao\\
  CNRS, ENS de Lyon, LIP\\
  Lyon, France\\
  michael.rao\imageat{}ens-lyon\imagedot{}fr \\
  \url{https://perso.ens-lyon.fr/michael.rao/}
\end{authorinfo}
\end{aicauthors}

\end{document}